\documentclass[journal,12pt,onecolumn,draftclsnofoot,]{IEEEtran}

\usepackage{amsthm}
\usepackage{amsmath}
\usepackage{graphicx,psfrag}
\usepackage{color}
\usepackage{textcase}
\usepackage[ruled]{algorithm2e}
\usepackage{subcaption}
\DeclareCaptionFormat{hfillstart}{\hfill#1#2#3\par}
\DeclareCaptionTextFormat{up}{\MakeTextUppercase{#1}}
\usepackage[cmintegrals]{newtxmath}

\usepackage[noadjust]{cite}

\newcommand{\bs}[1]{\boldsymbol{#1}}
\newcommand{\bm}[1]{\textbf{#1}}
\newcommand{\rev}{\backslash}
\newcommand{\sset}[1]{\mathcal{#1}}

\newtheorem{theorem}{\bm{Theorem}}
\newtheorem{mydef}{\bm{Definition}}
\newtheorem{prop}{\bm{Proposition}}

\newcommand{\inlineeqno}[1]{\;\;\;\refstepcounter{equation} (\theequation)\label{#1}}
\begin{document}

\title{Near-Optimal Sparse Sensing for Gaussian Detection with Correlated Observations}

\author{Mario~Coutino,~\IEEEmembership{Student Member,~IEEE,}
        Sundeep~Prabhakar~Chepuri,~\IEEEmembership{Member,~IEEE,}
        and~Geert~Leus,~\IEEEmembership{Fellow,~IEEE}
\thanks{This work is part of the ASPIRE project (project 14926 within the STW OTP programme), which is financed by the Netherlands Organisation for Scientific Research (NWO). Mario Coutino is partially supported by CONACYT. 

All the authors are with the Faculty of Electrical, Mathematics and Computer Science, Delft University of Technology, Delft 2628CD, The Netherlands (e-mail: m.a.coutinominguez@tudelft.nl; s.p.chepuri@tudelft.nl; g.j.t.leus@tudelft.nl).}
}

\markboth{Coutino et al. Near-Optimal Sparse Sensing for Gaussian Detection with Correlated Observations}%
{Shell \MakeLowercase{\textit{et al.}}: Bare Demo of IEEEtran.cls for IEEE Communications Society Journals}

\maketitle

\begin{abstract}
Detection of a signal under noise is a classical signal processing problem. When monitoring spatial phenomena under a fixed budget, i.e., either physical, economical or computational constraints, the selection of a subset of available sensors, referred to as sparse sensing, that meets both the budget and performance requirements is highly desirable. Unfortunately, the subset selection problem for detection under dependent observations is combinatorial in nature and suboptimal subset selection algorithms must be employed. In this work, different from the widely used convex relaxation of the problem, we leverage submodularity, the diminishing returns property, to provide practical near-optimal algorithms suitable for large-scale subset selection. This is achieved by means of low-complexity greedy algorithms, which incur a reduced computational complexity compared to their convex counterparts.
\end{abstract}


\IEEEpeerreviewmaketitle

\section{Introduction}
\IEEEPARstart{L}{arge} sensor networks are becoming pervasive in our daily life. They are found in monitoring activities, e.g, traffic flow and surveillance, as well as typical signal processing applications such as radar and seismic imaging. The data generated by these networks requires to undergo several processing steps before being used for inference tasks, such as estimation or detection. Due to the increase in the size of the network, managing the data throughput can become a challenging problem in itself. Hence, if a known inference task with fixed performance requirements is kept in mind during the design phase of a sampler, large data reduction benefits can be obtained by optimizing the number of deployed sensors. In realistic setups, the budget for a particular measurement campaign is also constrained, e.g., limited processing power, reduced hardware costs, and physical space restrictions. Therefore, it is of great importance to only deploy the sensors that provide meaningful information to solve the problem at hand. However, there is always a trade-off between the performance and the sparsity of the deployed network when such constraints are enforced. This framework in which a reduced number of sensors is employed for data acquisition is here referred to as \emph{sparse sensing}.

In this work, we are interested in the task of designing structured \emph{sparse samplers} for detecting signals under correlated measurements. In particular, we focus on the detection problem for the case of Gaussian measurements with non-diagonal covariance matrices. Such problems are commonly found in practical applications such as sonar and radar systems~\cite{VT1}, imaging~\cite{seis}, spectrum sensing for cognitive radio~\cite{CR}, and biometrics~\cite{BO}, to list a few. For this purpose, we consider a detection task in which a series of measurements, acquired in a distributed fashion, are gathered at a fusion center, e.g., the main processing unit, to perform a hypothesis test. We restrict ourselves to a binary decision problem, in which the fusion center has to decide between two available states $\{\mathcal{H}_{0},\mathcal{H}_{1}\}$ given the observed data. Following the conventional detection theoretical approach, we provide sparse sampler design strategies for both the Neyman-Pearson and Bayesian  setting. Furthermore, as our main goal is to provide a general and scalable framework capable of dealing with large-scale problems, we focus our attention on fast and near-optimal optimization techniques levering submodularity and greedy heuristics. This approach differs from the current state-of-the-art that is fundamentally based on convex relaxations. Despite the fact that the typical convex relaxations provide approximate solutions to the sensor selection problem, they boil down to semidefinite programs. These problems, albeit being solvable efficiently, are computationally intensive and, for large datasets, do not scale very well. 

\subsection{Prior Art}

The structured sparse sampler design problem consists of selecting the subset of measurements with the smallest cardinality possible such that some prescribed performance requirements are met. This problem is commonly referred to in the literature as sparse sensing or sensor selection~\cite{SS0}. Extensive research has been carried out in the area of sparse sensing for estimation~\cite{SS1}-\cite{SS4} and detection~\cite{P1}-\cite{sub:r15} problems. However, much of this work depends on the convex optimization machinery for optimizing the performance metrics or their surrogates. Current efforts, spanning from the field of operational research and machine learning, have shown that greedy heuristics provide near-optimal solutions, given that the cost to optimize satisfies certain properties~\cite{sub:r10,Sub2}. For these setups, sparse sensing has mostly been studied for estimation purposes, using information theoretic measures such as entropy and mutual information as well as experiment design metrics~\cite{Sub1}-\cite{ssVet3}, which exhibit the property of submodularity~\cite{sub:r4}. Similar to convex/concave functions, submodular set functions have certain properties that make them attractive for optimization. Particularly, under certain conditions, some problems involving submodular set functions allow linear-time optimization algorithms (in the size of the input set).
This fact has been fundamental for designing greedy sampling strategies in large scale problems~\cite{gredNew1}-\cite{sub:r12},~\cite{sub:lchamon1}. 

For the particular case of the detection task, the state-of-the-art structured sparse sampler design framework~\cite{P1}-\cite{sub:r15} aims to optimize surrogate functions of the probability of error for the case of binary hypothesis testing. For Gaussian processes with uncorrelated errors, the sampling problem can be solved \emph{optimally} in linear time. These optimal solutions are possible as it can be shown that maximizing the divergence measures between the probability distributions~\cite{sub:bhatt}, e.g., Kullback-Leibler (KL) divergence, $J$-divergence, or Bhattacharyya distance, is tantamount to optimizing the probability of error~\cite{sub:r15}.
However, when correlated errors are considered, optimizing the divergence measures is not exactly equivalent to optimizing the probability of error. Therefore, only suboptimal solutions can be obtained by maximizing the divergences. Furthermore, even though such divergences are simpler to optimize than the actual error probabilities, the problem remains non-convex in the selection variables. As a result, convex approximations must be performed in order to solve the sensor selection problem, often leading to a semidefinite program. However, despite of being solvable efficiently, these semidefinite programs  are not suitable for large-scale settings where our work takes the greatest interest.

\subsection{Overview and Main Contributions}
We concentrate on fast and near-optimal sparse sampler design for Gaussian detection problems with correlated errors. The typical surrogates for the \emph{probability of miss detection} $P_{\rm{m}}$ in the Neyman-Pearson setting, and the \emph{probability of error,} $P_{e}$, in the Bayesian setting, which are based on divergence measures between the two distributions, are in this work relaxed to provide submodular alternatives capable to tackle the sparse sampler design for large-scale problems. 

The main idea behind this work is to show, that in certain situations, it might be possible to avoid the convex machinery~\cite{sub:r15} to solve the sensor selection problem for detection. This becomes important when large scales are considered and fast algorithms are highly desirable. Unfortunately, as the submodular machinery is mostly suitable for problems where constraints can be described as matroids, i.e., structures that generalize the notion of linear independence in vector spaces, instances with a different kind of constraint, e.g., non-monotone separable constraints, cannot be properly addressed. Therefore, in this work, we mainly focus on cardinality constrained problems. In the following, our main contributions are highlighted.

\begin{itemize}
\item[--] For Gaussian observations with common covariance and distinct means we derive a bound for the approximate submodularity of the signal-to-noise ratio (SNR) set function, which provides grounds for the direct application of a greedy heuristic to maximize this cost set function under certain conditions. For instances where the near-optimality guarantees are weak, we derive a submodular set function surrogate based on the Schur complement. While this surrogate establishes a link with traditional convex relaxations for sparse sensing, it accepts a near-optimal maximization using a greedy algorithm that scales linearly in the number of selected sensors through its recursive description. This method attains results comparable with the ones of convex relaxation, but at a significantly reduced computational complexity. 
\item[--] For Gaussian observations with uncommon covariances and common means we show that the divergences between probability distributions are not submodular. Despite this, we present them as a difference of submodular functions, which can be approximately optimized. In cases where these decompositions are not readily available, we introduce surrogate decompositions based on the Schur complement. This approach provides local optimality guarantees without involving computationally intensive semidefinite programs as in the convex case. 
\end{itemize}

\subsection{Outline and Notation}
The rest of this paper is organized as follows. In Section~\ref{sec:2}, the problem of sparse sampler design for detection is introduced, and the sensor selection metrics for both the Neyman-Pearson and Bayesian setting are discussed. The submodular optimization theory is introduced in Section~\ref{sec:pre}. In Section~\ref{sec:ssum} and Section~\ref{sec:ssuc}, submodular set function surrogates for the selection criteria are derived and a general framework to solve the sparse sampler design for Gaussian measurements is provided. Finally, conclusions are drawn in Section~\ref{sec:conc}.

The notation used in this paper is the following. Upper (lower) bold faces letters are used to define matrices (column vectors). $\mathcal{N}(\cdot,\cdot)$ is reserved to represent a Gaussian normal distribution. The notation $\sim$ is read as ``is distributed according to". $(\cdot)^{T}$ and $(\cdot)^{-1}$ represent transposition and matrix inversion, respectively. $\text{diag}(\cdot)$ refers to a diagonal matrix with its argument on the main diagonal. $\bm{I}$ and $\bm 1$ denote the identity matrix and the all-one vector of appropriate size, respectively. $\det(\cdot)$ and $\log(\cdot)$ are the matrix determinant and natural logarithm, respectively. $\text{tr}\{\cdot\}$ denotes the matrix trace operator. $[\bm x]_i$ and $[\bm X]_{i,j}$ denote the $i$th entry of the vector $\bm x$ and the $(i,j)$ entry of the matrix $\bm X$, respectively. Calligraphic letters denote sets, e.g., $\mathcal{A}$, and the vector $\bm 1_{\sset{A}}$, with $\sset{A}\subseteq\sset{V}$, denotes a vector with ones at the indices given by $\mathcal{A}$, and zeros in the complementary set, $\sset{V}\setminus\sset{A}$.

\section{Problem Statement}
\label{sec:2}
Consider a set $\mathcal{X}=\{x_{1},\ldots,x_{M}\}$ of $M$ candidate measurements. These measurements can be temporal samples of temperature, spatial samples from wavefield measurements, etc. The samples are known to be related to the models
\begin{eqnarray}
\mathcal{H}_{0} : x_{m} \sim p_{m}(x\vert \mathcal{H}_{0}),\; m = 1,2,\ldots, M, \\
\mathcal{H}_{1} : x_{m} \sim p_{m}(x\vert \mathcal{H}_{1}),\; m = 1,2,\ldots, M,
\label{eq:detprob}
\end{eqnarray}
where $p_m(x\vert\mathcal{H}_i)$ for $i=0,1$ denotes the probability density function (pdf) of the $m$th measurement, $x_m$, conditioned on the state $\mathcal{H}_i$. By stacking the elements of $\mathcal{X}$ in a vector $\textbf{x}=[x_1,x_2,\ldots,x_m]^T\in\mathbb{R}^{M}$, the pdf of the measurement set for the hypothesis $\mathcal{H}_i$ is denoted by $p(\textbf{x}\vert\mathcal{H}_i)$.
\begin{table*}[th]
\caption{Summary of divergence measures for Gaussian probability distributions.}
\label{tab:table3}
\centering
\begin{tabular}{ccc}
\hline
\centering
Divergence & Expression  & Setting\\
\hline
& Bhattacharyya & \\
 $\sset{B}(\sset{H}_1\Vert\sset{H}_0) \coloneq $ & $ 
\frac{1}{8}(\bs{\theta}_{1}-\bs{\theta}_{0})^{T}\bs{\Sigma}^{-1}(\bs{\theta}_{1}-\bs{\theta}_{0}) + \frac{1}{2}\log\bigg(\frac{\det(\bs{\Sigma})}{\sqrt{\det(\bs{\Sigma}_{1})\det(\bs{\Sigma}_{0})}}\bigg),\; \bs\Sigma = 0.5(\bs\Sigma_1 + \bs\Sigma_0)$
\inlineeqno{eq:bhatDiv1} & Bayesian\\
\hline
& Kullback-Leibler & \\
$\sset{K}(\sset{H}_1\Vert\sset{H}_0) \coloneq $ & $\frac{1}{2}\bigg(\text{tr}\big(\bs{\Sigma}_{0}^{-1}\bs{\Sigma}_{1}\big) + \big(\bs{\theta}_{1}-\bs{\theta}_{0}\big)^{T}\bs{\Sigma}_{0}^{-1}\big(\bs{\theta}_{1}-\bs{\theta}_{0}\big) - N + \log\big(\det(\bs{\Sigma}_{0})\big) - \log\big(\det(\bs{\Sigma}_{1})\big)
\bigg)
\inlineeqno{eq:kldivdi}$ & Neyman-Pearson\\
\hline
& J-Divergence & \\
$\sset{D}_{\rm{J}}(\sset{H}_0\Vert\sset{H}_1)\coloneq$ & $\mathcal{K}(\mathcal{H}_{1}\Vert\mathcal{H}_{0}) +  \mathcal{K}(\mathcal{H}_{0}\Vert\mathcal{H}_{1})\inlineeqno{eq:jedi}$ & Neyman-Pearson
 \\ \hline
\end{tabular}
\end{table*}

We pose the acquisition of a reduced set $\mathcal{Y}\subseteq\mathcal{X}$ consisting of $K$ measurements as a linear sensing problem where the rows of the sensing matrix yield a subset of rows of an identity matrix. The selected rows, indexed by $\mathcal{A}$, of the identity matrix are defined by a vector $\textbf{w}$ whose entries belong to a binary alphabet set, i.e.,
\begin{equation}
\textbf{w} = [w_{1},w_{2},\ldots,w_{M}]^{T}\in\{0,1\}^{M},
\end{equation}
where $w_{m} = 1\,(0)$ indicates that the $m$th measurement is (not) selected. The subset of rows is then defined as
\begin{equation}
\mathcal{A} := \{m\; |\; w_{m} = 1, 1 \leq m \leq M\}.
\end{equation}
The acquisition scheme can be formally  expressed using the following linear model
\begin{equation}
\textbf{y}_\mathcal{A} = \bs{\Phi}_\mathcal{A}\textbf{x} \in \mathbb{R}^{K},
\label{eq:linModelY}
\end{equation}
where $\textbf{y}_{\mathcal{A}} = [y_{1},y_{2},\ldots,y_{K}]^T$ is the reduced-size measurement vector whose entries belong to the set $\mathcal{Y}\subseteq\mathcal{X}$. The selection matrix $\boldsymbol{\Phi}_{\mathcal{A}}$ is a binary matrix composed of the rows of the identity matrix defined by the set $\mathcal{A}$ (non-zero entries of $\textbf{w}$). Even though $K$ is (possibly) unknown to us, we are interested in cases where $K \ll M$, as it is desirable to perform inference on a reduced measurement set. As the notation based on either $\bm{w}$ or $\mathcal{A}$ is considered interchangeable, from this point on, we make no distinction between them.

The subset of measurements $\mathcal{Y}$ is finally used to solve the detection problem~\eqref{eq:detprob} given that the detection performance requirements, for a given application, are met. If the prior hypothesis probabilities are known (Bayesian setting), the optimal detector minimizes the probability of error,
$P_{\rm e} = P(\mathcal{H}_{0}|\mathcal{H}_{1})P(\mathcal{H}_{1}) + P(\mathcal{H}_1|\mathcal{H}_{0})P(\mathcal{H}_{0})$,
where $P(\mathcal{H}_{i}|\mathcal{H}_{j})$ is the conditional probability of deciding $\mathcal{H}_{i}$ when $\mathcal{H}_{j}$ is true and $P(\mathcal{H}_{i})$ is the prior probability of the $i$th hypothesis. For unknown prior hypothesis probabilities (Neyman-Pearson setting), the optimal detector aims to minimize the probability of miss detection (type II error),
$P_{\rm m} = P(\mathcal{H}_{0}|\mathcal{H}_{1}),$
for a fixed probability of false alarm (type I error),
$P_{\rm{fa}} = P(\mathcal{H}_{1}|\mathcal{H}_{0}).$

In a more formal manner, the sensor selection problem for detection, in both settings, is given by
\begin{align}
&\text{{\bf Bayesian}:} \quad \quad \text{arg min}_{\mathcal{A}} \;\; P_{\rm e}(\mathcal{A}) \quad \text{s. to} \; \;|\mathcal{A}| = K, 
\label{eq:intp1} \\
&\text{{\bf Neyman-Pearson}:} \quad \quad  \text{arg min}_{\mathcal{A}}  \;\;P_{\rm m}(\mathcal{A}) \nonumber\\
&\hskip3.4cm\text{s. to} \; |\mathcal{A}| = K, \;\; P_{\rm{fa}}(\mathcal{A}) \leq \lambda, 
\label{eq:intp2}
\end{align}
where $P_{\rm e}(\mathcal{A})$, $P_{\rm m}(\mathcal{A})$ and $P_{\rm{fa}}(\mathcal{A})$ denote the error probabilities due to the measurement selection defined by the set $\mathcal{A}$, and $\lambda$ the prescribed accuracy of the system. 

As for the most general case, the performance metrics in~\eqref{eq:intp1} and~\eqref{eq:intp2} are not easy to optimize numerically, we present alternative measures that can be used as direct surrogates to solve the optimization problems ~\eqref{eq:intp1} and~\eqref{eq:intp2}.
Here, we focus on metrics which provide a notion of distance between the hypotheses under test. That is, we are interested in \emph{maximizing} the distance between two distinct probability distributions $p(\bm{y}_{\mathcal{A}}|\mathcal{H}_{i})$ and $p(\bm{y}_{\mathcal{A}}|\mathcal{H}_{j})$ using a divergence measure $\sset{D}(p(\bm{y}_{\mathcal{A}}|\mathcal{H}_{i})\Vert p(\bm{y}_{\mathcal{A}}|\mathcal{H}_{j}))\in\mathbb{R}_{+}$. They lead to tractable optimization methods and, in some particular cases such as for independent observations under uncorrelated Gaussian noise, they result in an optimal solution. A summary of the divergences, for Gaussian probability distributions, $\mathcal{N}(\bs{\theta}_i,\bs{\Sigma}_i)$, between the different hypotheses under test employed in this work is shown in Table~\ref{tab:table3}. Here, $\bs{\theta}_i$ and $\bs{\Sigma}_i$ denote the mean vector and the covariance matrix of the $i$th distribution, respectively. For a more detailed treatment of these divergence measures and their suitability for sensor selection, the reader is referred to~\cite{SS0},~\cite{sub:shannon}-\cite{sub:kullbck} and the references therein.

Using these divergence measures, the relaxed formulation of the sparse sensing problems~\eqref{eq:intp1} and~\eqref{eq:intp2} can be stated, respectively, as cardinality constraint (P-CC) and detection performance constraint (P-DC) problems:
\begin{align}
&\text{{\bf P-CC}:} \quad \quad {\text{arg max}}_\mathcal{A} \;\; f(\mathcal{A}) \quad \text{s. to}  \;\;|\mathcal{A}| = K; 
\label{eq:p1} \\
&\text{{\bf P-DC}:} \quad \quad  {\text{arg min}}_\mathcal{A} \;\;\;|\mathcal{A}| \quad \text{s. to} \;\; f(\mathcal{A}) \geq \lambda, 
\label{eq:p2}
\end{align}
%
%
where $f(\mathcal{A})$ is one of the divergence measures, $\lambda$ is the prescribed accuracy and $K$ is the cardinality of the selected subset of measurements. 
For the sake of exposition, in this paper we mainly focus on cardinality constraints (i.e., a uniform matroid constraint). However, the methods presented in this work can be easily extended to budget functions, expressed as constraints, representable by other kinds of matroids~\cite{matroid}.
\section{Preliminaries}
\label{sec:pre}
In this section, some preliminaries about submodularity are provided. The main definitions and theorems related to submodular set functions use throughout the work are presented. 
\subsection{Submodularity}
In many engineering applications we encounter the \emph{diminishing returns} principle. That is, the gain of adding new information, e.g., a data measurement, to a large pool of measurements is smaller than the gain of adding the same piece of information to a smaller pool of measurements. This notion is mathematically represented by the next definition.
\begin{mydef}{\textnormal{\bm{(Submodularity)}}} Let $\mathcal{V}=\{1,2,\ldots,M\}$ refer to a ground set, then $f:2^{\vert\mathcal{V}\vert}\rightarrow\mathbb{R}$ is said to be submodular, if for every $\mathcal{A}\subseteq\mathcal{B}\subset\mathcal{V}$ and $v\in\mathcal{V}\setminus\mathcal{B}$ it holds that
\begin{equation}
f(\mathcal{A}\cup\{v\}) - f(\mathcal{A}) \geq f(\mathcal{B}\cup\{v\}) - f(\mathcal{B}).
\label{eq:subDef}
\end{equation}
\end{mydef}

Similar to convex functions, submodular set functions have certain properties that make them convenient to optimize. For example, the unconstrained minimization of general submodular functions can be done in polynomial time~\cite{sub:r3} with respect to the size of the ground set $\vert\sset{V}\vert$.

Even though the maximization of general submodular set functions is an NP-hard problem, Nemhauser et al.~\cite{sub:r9} have shown that for the maximization of a non-decreasing submodular set function $f$, with $f(\emptyset) = 0$, the simple \emph{greedy} procedure presented in Algorithm~\ref{al:greedy} finds a solution which provides at least a constant fraction $(1-1/e)\approx 0.63\%$ of the optimal value. In this context, a set function $f:2^{\vert\sset{V}\vert}\rightarrow\mathbb{R}$ over a ground set $\sset{V}$ is considered non-decreasing if and only if $f(\sset{B}) \geq f(\sset{A})$ holds for all sets $\sset{A}\subseteq\sset{B}\subseteq\sset{V}$.

Using similar arguments, Krause et al.~\cite{sub:r10} extended the near-optimality of the greedy heuristic for \emph{approximately} submodular set functions or $\epsilon$-submodular set functions:
\begin{mydef}\textnormal{\bm{(}$\bs{\epsilon}$-\bm{Submodularity)}~\cite{sub:r10}} A set function $f:2^{\vert\sset{V}\vert}\rightarrow\mathbb{R}$ defined over a ground set $\sset{V}$, is approximately submodular with constant $\epsilon$ or $\epsilon$-submodular, if for all sets $\sset{A}\subseteq\sset{B}\subset{V}$, and $v\in\sset{V}\rev\sset{B}$ it holds that
\begin{equation}
f(\sset{A}\cup\{v\}) - f(\sset{A}) \geq f(\sset{B}\cup\{v\}) - f(\sset{B}) - \epsilon.
\end{equation}
\end{mydef}
For $\epsilon$-submodular functions the greedy Algorithm~\ref{al:greedy} provides the following weaker guarantee.
\begin{theorem}\textnormal{\bm{(}$\bs{\epsilon}$-\bm{Near-Optimality)}}~\normalfont{\cite{sub:r10}} Let $f:2^{\vert\sset{V}\vert}\rightarrow\mathbb{R}$ be a normalized, i.e., $f(\emptyset) = 0$, non-decreasing, $\epsilon$-submodular set function defined over a finite ground set $\sset{V}$. Let $\sset{G}$ be the set of $K$ elements chosen by Algorithm~\ref{al:greedy}. Then,
\begin{equation}
f(\sset{G}) \geq \bigg(1 - \frac{1}{e}\bigg)f(\mathcal{A}_{\rm{opt}})-K\epsilon,
\label{eq:e}
\end{equation}
where $\mathcal{A}_{\rm{opt}} := \underset{\sset{A}\subseteq\sset{V},\vert\sset{A}\vert=K}{\text{arg max}}\;f(\sset{A})$ is the optimal set.
\label{th:e}
\end{theorem}
\setlength{\textfloatsep}{5pt}
\begin{algorithm}[t]
\caption{\textsc{Greedy Algorithm.}}
\SetAlgoLined
\KwResult{$\sset{A} : |\sset{A}| = K$}
 initialization $\sset{A} = \emptyset$, $k = 0$\;
 \While{$k < K$}{
  $a^{*} = \underset{a\not\in \sset{A}}{\text{arg max}}\;f(\sset{A}\cup \{a\})$\;
  $\sset{A} = \sset{A}\cup \{a^{*}\}$\;
  $k = k + 1$\;
 }
 \label{al:greedy}
\end{algorithm}

The result in Theorem~\ref{th:e} implies that for small $K\epsilon$, Algorithm~\ref{al:greedy} provides a good approximate solution for the maximization under cardinality constraints. As in practice it is observed that the lower bound from~\cite{sub:r9} is not tight, i.e., the greedy method performs much better than the lower bound~\cite{sub:lchamon1}, the expression provided for $\epsilon$-submodular set functions in~\eqref{eq:e} is expected to be also a loose bound for the performance of Algorithm~\ref{al:greedy}. In any case, Theorem~\ref{th:e} shows that the degradation on the approximation factor increases as $K$ becomes larger.
\subsection{Difference of Submodular Functions}
A notable result in combinatorial optimization arises from the fact that {\it any} set function can be expressed as a difference of two submodular set functions~\cite{sub:r11}. Therefore, the optimization problem
\begin{equation}
\underset{\sset{A}\subseteq\sset{V}}{\text{max}}\;f(\sset{A}) \equiv \underset{\sset{A}\subseteq\sset{V}}{\text{max}}\;[g(\sset{A}) - h(\sset{A})],
\label{eq:ds_max}
\end{equation}
where the cost set function $f:2^{\vert\sset{V}\vert}\rightarrow\mathbb{R}$ is expressed as the difference of two set functions $g:2^{\vert\sset{V}\vert}\rightarrow\mathbb{R}$ and $h:2^{\vert\sset{V}\vert}\rightarrow\mathbb{R}$, defined over a ground set $\sset{V}$ is, in general, NP-hard. Recent results from Iyer et al.~\cite{sub:r12} show that the general case of this problem is multiplicatively inapproximable. However, in this work we motivate the usage of practical methods, employing well-designed heuristics, to obtain good results when solving large-scale real-world problems.

Firstly, let us consider a heuristic from convex optimization for approximating the problem of minimizing the difference of convex functions. A typical heuristic is obtained by linearizing one of the convex functions with its Taylor series approximation. With such a linearization, the original nonconvex minimization problem can be transformed into a sequential minimization of a convex plus an affine function. In the literature this method is known as the convex-concave procedure (CCP)~\cite{sub:r13}. Similarly, for maximizing the difference of submodular set functions, it is possible to substitute one of the submodular set functions from~\eqref{eq:ds_max} by its modular upper bound at every iteration as suggested in~\cite{sub:r12}. Algorithm~\ref{al:supsub} summarizes the supermodular-submodular (SupSub) procedure as described in~\cite{sub:r12} when the cardinality of the set is constrained for approximating the solution of~\eqref{eq:ds_max}.
\setlength{\textfloatsep}{5pt}
\begin{algorithm}[t]
 \caption{\textsc{SupSub Procedure}}
\SetAlgoLined
\KwResult{$\sset{A}$}
 initialization $\sset{A}^{0}=\emptyset;\;t = 0$\;
 \While{\text{not converged (i.e., $(\sset{A}^{t+1} \neq \sset{A}^t)$)}}{
	$\sset{A}^{t+1} := \underset{\sset{A},\;\vert\sset{A}\vert = k}{\text{argmax}}\;g(\sset{A})-m_{\sset{A}^t}^{h}(\sset{A})$\;
    $t = t+1$\;
 }
 \label{al:supsub}
\end{algorithm}

In Algorithm~\ref{al:supsub}, at every iteration, a submodular set function is maximized. This is due to the fact that the modular upper bound $m_{\sset{A}^t}^{h}$ of $h$, locally to $\sset{A}^t$~\cite{sub:r9}, preserves the submodularity of the cost. Using the characterization of submodular set functions two tight modular upper bounds can be defined as follows
\begin{align}
m_{\sset A,1}^{h}(\sset{C})\triangleq h(\sset{A}) &- \sum\limits_{j\in\sset{A}\setminus\sset{C}}h(\{j\}\vert\sset{A}\setminus\{j\}) \nonumber\\
&\quad \quad \quad + \sum\limits_{j\in\sset{C}\setminus\sset{A}}h(\{j\}\vert\emptyset),\label{eq:modBndA}\\
m_{\sset A,2}^{h}(\sset{C}) \triangleq h(\sset{A}) &- \sum\limits_{j\in\sset{A}\setminus\sset{C}}h(\{j\}\vert\sset{V}\setminus\{j\}) \nonumber \\
&\quad\quad\quad+ \sum\limits_{j\in\sset{C}\setminus\sset{A}}h(\{j\}\vert\sset{C})\label{eq:modBndB},
\end{align}
where $h(\sset{A}\vert\sset{C}) \triangleq h(\sset{C}\cup\sset{A}) - h(\sset{C})$ denotes the gain of adding $\sset{A}$ when $\sset{C}$ is already selected. In practice, either~\eqref{eq:modBndA} or~\eqref{eq:modBndB} can be employed in Algorithm~\ref{al:supsub} or both can be run in parallel choosing the one that is better. For a more in-depth treatment of these bounds, the reader is referred to~\cite{subbnd}. These bounds follow similar arguments as the ones found in majorization-minimization algorithms~\cite{mmAlg} for general non-convex optimization.

 Although the maximization of submodular functions is NP-hard, Algorithm~\ref{al:greedy} can be used to approximate at each step the maximum of the submodular set function in Algorithm~\ref{al:supsub}. Furthermore, as the problem of submodular maximization with cardinality, matroid and knapsack constraints admits a constant factor approximation, the SupSub procedure can be easily extended to constrained minimization of a difference of two submodular functions. In addition, similar to CCP, the SupSub procedure is guaranteed to reach a local optimum of the set function when the procedure converges~\cite{sub:r12}, i.e., $\sset{A}^{t+1} = \sset{A}^t$.
 
The main reasons to prefer the SupSub procedure over, a possibly, submodular-supmodular (SubSup) procedure, where a modular lower bound of $g(\cdot)$ is used and the inner step consists of the minimization of a submodular function, are its computationally complexity and versatility. Even though unconstrained minimization of submodular set functions can be performed in polynomial time, the addition of constraints to the minimization of submodular set functions renders the problem NP-hard, for which there are no clear approximation guarantees. As a result, the SupSub is often preferred for optimizing differences of submodular functions.

\section{Observations with Uncommon Means}
\label{sec:ssum}
In this section, we illustrate how to design sparse samplers using the criteria presented in Section~\ref{sec:2} for Gaussian observations with uncommon means. This kind of measurements arises often in communications as in the well-studied problem of detecting deterministic signals under Gaussian noise.

Consider the binary signal detection problem in~\eqref{eq:detprob}. Furthermore, let us assume that the pdfs of the observations are multivariate Gaussians with uncommon means and equal covariance matrices. Then, the related conditional distributions, under each hypothesis, are given by
\begin{equation}
\begin{array}{l}
\sset{H}_{0} : \bm{y}_{\mathcal{A}}\sim\sset{N}(\bm{0},\bs{\Sigma}_{\mathcal{A}})\\
\sset{H}_{1} : \bm{y}_{\mathcal{A}}\sim\sset{N}(\bs{\theta}_{\mathcal{A}},\bs{\Sigma}_{\mathcal{A}}),
\end{array}
\label{eq:hp2}
\end{equation}
where $\sset{A}\subseteq\sset{V}$ is the subset of selected sensors from the set of candidate sensors $\sset{V}=\{1,2,\ldots,M\}$, and where $\bs{\theta}_{\sset{A}} = \bs{\Phi}_\sset{A}\bs{\theta}\in\mathbb{R}^{K}$ and $\bs{\Sigma}_{\sset{A}} = \bs{\Phi}_\sset{A}\bs{\Sigma}\bs{\Phi}_\sset{A}^{T}\in\mathbb{R}^{K\times K}$. The mean vector $\bs{\theta}$ and the covariance matrix $\bs{\Sigma}$ are assumed to be known a priori.

By observing the Bhattacharyya distance and the KL divergence in~\eqref{eq:bhatDiv1} and~\eqref{eq:kldivdi}, respectively, it can be seen that for the probability distributions in~\eqref{eq:hp2} such metrics are reduced to the so-called signal-to-noise ratio function
\begin{equation}
s(\sset{A}) = \bs{\theta}_\sset{A}^T\bs{\Sigma}^{-1}_\sset{A}\bs{\theta}_\sset{A}.
\label{eq:csnr}
\end{equation}
Therefore, maximizing the signal-to-noise ratio, $s(\sset{A})$, directly maximizes the discussed divergence measures leading to an improvement in the detection performance. As a result, we are required to solve the following combinatorial problem
\begin{equation}
\underset{\sset{A}\subseteq\sset{V};\vert\sset{A}\vert=K}{\text{arg max}}\;s(\sset{A})
\label{eq:snr}
\end{equation}

Due to the hardness of the problem in~\eqref{eq:snr}, finding its exact solution requires an exhaustive search over ${M}\choose{K}$ possible combinations which for large $M$ rapidly becomes intractable. Simplifications for the problem~\eqref{eq:snr} can be derived using convex optimization~\cite{sub:r15}. Such approaches provide a sub-optimal solution in polynomial time when cast as a semidefinite program (SDP). Even though under the SDP framework, approximate solutions for~\eqref{eq:snr} can be found efficiently, for large-scale problems near-optimal solutions obtained through Algorithm~\ref{al:greedy} are more attractive as they scale linearly in the number of selected sensors, and only depend on the efficient evaluation of the cost function.

\subsection{$\epsilon$-Submodularity of Signal-to-Noise Ratio}
Although the signal-to-noise ratio set function is not submodular, we can try to quantify how far this set function is away from being submodular. For this purpose, we derive a bound for the $\epsilon$-submodularity of the signal-to-noise ratio.

In the following, we present a key relationship between the parameter $\epsilon$ and the conditioning of the covariance matrix $\bs{\Sigma}$ to provide a bound on the approximate submodularity of the signal-to-noise ratio set function. This relation is summarized in the following theorem.

\begin{theorem}
\label{th:epSub}
If $\bs{\Sigma}$ has a condition number $\kappa\coloneq\lambda_{\max}/\lambda_{\min}$, minimum eigenvalue $\lambda_{\min}$, maximum eigenvalue $\lambda_{\max}$, and admits a decomposition $\bs\Sigma = a\textbf{I} + \textbf{S}$ where $a = \beta\lambda_{\min}$ with $\beta\in(0,1)$, then the signal-to-noise ratio set function $s(\mathcal{A})$ is $\epsilon$-approximately submodular with
\begin{equation}
\epsilon \leq 4C_{1}\bigg(a+\frac{\kappa\lambda_{\max}}{\beta}\bigg),
\label{eq:thEpSub}
\end{equation}
where 
$C_{1} = \Vert\bm{S}^{-1}\bs{\theta}\Vert^2$, with $\bs{\theta}$ being the mean vector.
\end{theorem}
\begin{proof}
See Appendix~\ref{ap:esub}.
\end{proof}

From the result in~\eqref{eq:thEpSub}, it is clear that when the condition number of the covariance matrix $\bs{\Sigma}$ is low, the theoretical guarantee in Theorem~\ref{th:e} provides encouraging bounds for the greedy maximization of the signal-to-noise ratio. In this regime, several works have focused on sensor selection in the past. For example, in the limiting case $\bs{\Sigma} = \sigma^2\bm{I}$, where $s(\mathcal{A})$ becomes a modular set function (the expression in~\eqref{eq:subDef} is met with equality), it has been shown that the optimization problem can be solved optimally by sorting~\cite{sub:r15}. Unfortunately, for arbitrary covariance matrices (especially badly conditioned matrices), the $\epsilon$-submodular guarantee can be very loose. In that case, surrogate submodular set functions can be efficiently optimized using Algorithm~\ref{al:greedy} as a fast alternative for performing sensor selection in large-scale problems.

\subsection{Signal-to-Noise Ratio Submodular Surrogate}
Firstly, let us express the covariance matrix $\bs{\Sigma}$ as
\begin{equation}
\bs{\Sigma} = a\bm{I} + \bm{S},
\label{eq:sigma}
\end{equation}
where $a\in\mathbb{R}$ and $\bm S\in\mathbb{R}^{M\times M}$ have been chosen as described in Theorem~\ref{th:epSub}. Combining~\eqref{eq:csnr} and~\eqref{eq:sigma}, it can be shown that the signal-to-noise ratio can be rewritten as~\cite{sub:r15}
\begin{equation}
s(\sset{A}) = \bs{\theta}^{T}\textbf{S}^{-1}\bs{\theta} - \bs{\theta}^{T}\textbf{S}^{-1}\big[\textbf{S}^{-1} + a^{-1}\text{diag}(\bm 1_{\sset{A}})\big]^{-1}\textbf{S}^{-1}\bs{\theta},
\label{eq:snrExt}
\end{equation}
where the non-zero entries of the vector $\bm 1_{\sset{A}}$ are given by the set $\sset{A}$. Then, considering that the signal-to-noise ratio is always non-negative we can use the Schur complement to express this condition as a linear matrix inequality (LMI) in $\bm{w}$,
\begin{equation}
\textbf{M}_{\sset{A}} \coloneqq \begin{bmatrix}
\textbf{S}^{-1} + a^{-1}\text{diag}(\bm 1_{\sset{A}}) & \textbf{S}^{-1}\bs{\theta}\\
\bs{\theta}^{T}\textbf{S}^{-1} & \bs{\theta}^{T}\textbf{S}^{-1}\bs{\theta} 
\end{bmatrix} \succeq 0,
\label{eq:lmisnr}
\end{equation}
which is similar to the LMI found in the convex program in~\cite{sub:r15}. Therefore, we can consider the following optimization problem as an approximation of~\eqref{eq:snr}
\begin{equation}
\begin{aligned}
&\underset{\sset{A}\subseteq\sset{V};\vert\sset{A}\vert=K}{\text{arg max}} & & f(\sset{A})
\end{aligned}
\label{eq:subs2}
\end{equation}
where the cost set function has been defined as
\begin{equation}
f(\sset{A}) \triangleq 
\begin{cases}
	0, & \text{if } \sset{A} = \emptyset \\
	\log\det(\bm{M}_{\sset{A}}), & \text{if } \sset{A} \neq \emptyset
\end{cases}.
\label{eq:subcostdef}
\end{equation}
The normalization of the cost is done to avoid the infinity negative cost due to the logarithm of zero.

In the following, we motivate why~\eqref{eq:subs2} is a good alternative for~\eqref{eq:snr}. First, notice that the determinant of $\textbf{M}_{\sset{A}}$ consists of the product of two terms, where one of them is related to the signal-to-noise ratio $s(\sset{A})$. That is, using the generalization of the determinant for block matrices, we can decompose the determinant of the right-hand-side (RHS) of~\eqref{eq:lmisnr} as
\begin{align}
\text{det}(\textbf{M}_{\sset{A}}) &=\text{det}\begin{bmatrix}
\textbf{A} & \textbf{B}\\
\textbf{C} & \textbf{D}
\end{bmatrix} = \text{det}(\textbf{A})\text{det}(\textbf{D}-\textbf{C}\textbf{A}^{-1}\textbf{B})\\
&=\gamma(\sset{A}) s(\sset{A}),
\label{eq:prod}
\end{align}
where $\gamma(\sset{A}) = \text{det}(\textbf{S}^{-1}+a^{-1}\text{diag}(\bm 1_{\sset{A}}))$ with  $\gamma(\sset{A})> 0$.

From~\eqref{eq:prod} we notice that the determinant of $\textbf{M}_{\sset{A}}$ consists of the product of the signal-to-noise ratio $s({\sset{A}})$, and $\gamma({\sset{A}})$, which is inversely proportional to the loss in signal-to-noise ratio in~\eqref{eq:snrExt}. Furthermore, the determinant of $\textbf{M}_{\sset{A}}$ given by the product form in~\eqref{eq:prod} can be equivalently expressed as
\begin{equation}
\det(\textbf{M}_{\sset{A}}) = C - \bs{\theta}^{T}\bm{S}^{-1}\text{adj}\bigg(\bm{S}^{-1}+a^{-1}\text{diag}(\bm 1_{\sset{A}})\bigg)\bm{S}^{-1}\bs{\theta},
\end{equation}
where $\text{adj}(\bm{A})$ is the adjugate of $\bm{A}$ defined as the transpose of the cofactor matrix of $\bm{A}$, and $C$ is a constant term. Therefore, maximizing~\eqref{eq:prod} effectively maximizes a modified version of~\eqref{eq:snrExt} where the inverse of $\bm{S}^{-1} + a^{-1}\text{diag}(\bm 1_{\sset{A}})$  has been substituted by its adjugate.

In the following, we present a proposition that is required to provide guarantees for near optimality when the proposed submodular cost set function for the combinatorial problem~\eqref{eq:subs2} is maximized.

\begin{prop}\textnormal{\bm{(Monotonicity and Submodularity)}} The cost set function in~\eqref{eq:subcostdef} is a monotone non-decreasing submodular set function.
\label{prop:mon}
\label{prop:subs}
\end{prop}
\begin{proof}
See Appendix~\ref{ap:mono}.
\end{proof}
By the fact that the cost set function~\eqref{eq:subcostdef} is a normalized, non-decreasing submodular set function,~\eqref{eq:subs2} can be solved near-optimally for any cardinality size $K$ using Algorithm~\ref{al:greedy}. 

\subsection{Recursive Description of Cost Set Function} It is important to remark that most of the claims of scalability in submodular optimization works rely on the linear-time complexity with respect to the cardinality of the selected set. However, this claim might not translate in a fast optimization solver for all problem instances as the evaluation of the set function itself can be computationally expensive, and in certain situations, it is a prohibitive endeavor.

Under this perspective, we demonstrate the suitability of a large-scale optimization of~\eqref{eq:subs2} by showing that it is possible to compute this set function recursively, alleviating the complexity of computing the determinant of an $(M+1)\times (M+1)$ matrix, which in general, has complexity $\sset{O}((M+1)^{3})$. 

{
Let us consider the $k$th step of the greedy algorithm, with $\sset{A}_{k-1}$ denoting the subset of sensors selected upto this point. First recall that the set cost function~\eqref{eq:subcostdef} can be expressed as [cf.~\eqref{eq:prod}]
\begin{equation}
f(\sset{A}_{k}) = \ln(\det(\bm S^{-1} + a^{-1}\text{diag}(\bm 1_{\sset{A}_{k}}))s(\sset{A}_{k})).
\label{eq:costDet}
\end{equation}
By applying the determinant lemma to~\eqref{eq:costDet} we obtain
\begin{eqnarray}
f(\sset{A}_k)  &=& \ln(\det(\bm S^{-1})\det(\bm I + a^{-1}\bm S_{\sset{A}_{k}})s(\sset{A}_{k})),
\label{eq:costDet2}
\end{eqnarray}
where for $\sset{A}_{k} = \{m_1,\ldots,m_k\}$, we have defined $[\bm S_{\sset{A}_k}]_{i,j} = [\bm S]_{m_i, m_j}$. Here, $m_i$ is the  sensor index selected at the $i$th step.
As the $k$th step of the greedy algorithm evaluates the cost set function for the set $\sset{A}_k = \sset{A}_{k-1} \cup \{i\},\;\forall\;i\in\sset{V}\setminus\sset{A}$, in order to find the best sensor to add, the matrix in the second term of~\eqref{eq:costDet2} can be written using the following block structure
\begin{equation}
\bm I + a^{-1}\bm S_{\sset{A}_k} = \left[\begin{array}{c|c}
\bm I + a^{-1}\bm S_{\sset{A}_{k-1}} & \bm s_{\sset{A}_{k}} \\
\hline
\bm s_{\sset{A}_{k}}^T &  1+ \alpha_{\sset{A}_{k}}
\end{array}\right],
\end{equation}
where for $\sset{A}_{k-1}=\{m_1,\ldots,m_{k-1}\}$, we have defined $[\bm s_{\sset{A}_{k}}]_j = [\bm S]_{m_j,i}$, and $\alpha_{\sset{A}_{k}} = [\bm S]_{i,i}$. Therefore, using the property of the determinant for block matrices, we can construct the following recursive evaluation for the cost set function
\begin{equation}
\begin{split}
f(\sset{A}_k) &= \ln(\det(\bm S^{-1})\det(\bm I + a^{-1}\bm S_{\sset{A}_{k-1}})\\
&\times(1+\alpha_{\sset{A}_{k}} - \bm s_{\sset{A}_{k}}^T(\bm I + a^{-1}\bm S_{\sset{A}_{k-1}})^{-1}\bm s_{\sset{A}_{k}})s(\sset{A}_{k})),
\end{split}
\label{eq:recCost}
\end{equation}
where the matrix $\bm I + a^{-1}\bm S_{\sset{A}_{k-1}}$ is fixed for every $i\in\sset{V}\setminus\sset{A}_{k-1}$, and it only has to be updated when the sensor for the $k$th step has been chosen.

From~\eqref{eq:recCost}, the computational advantages during function evaluations are clearly seen. First, computation of the inverse of the matrix $\bm S^{-1}$ is not needed as for any set the term $\det(\bm S^{-1})$ is constant. This contrasts with the convex method from~\cite{sub:r15} which requires the inversion of $\bm S$. Second, the rank-one update of the inverse in~\eqref{eq:recCost} as well as the computation of $s(\sset{A}_k)$ have worst-case complexity $\sset{O}(K^2)$, which implies that the overall complexity of the proposed algorithm is about $\sset{O}(MK^3)$. That is, differently from its convex counterpart which has cubic complexity in the number of available sensors, the proposed method scales linearly with the number of available sensors. 

Furthermore, as seen in~\eqref{eq:recCost} it is possible to generate two solutions by the evaluation of the set cost function: $(i)$ the solution for maximizing greedily the submodular surrogate $f(\cdot)$, and $(ii)$ the solution of maximizing greedily the signal-to-noise ratio, $s(\cdot)$. Therefore, the proposed cost set function is perfectly suitable for large-scale problems, especially for instances with $M \gg K$, where computational complexity is of great importance. In addition, as two solutions can be built simultaneously, the one with the best performance can always be chosen as final solution. In addition, \emph{lazy evaluations}~\cite{ref:lazy}, or stochastic greedy selection~\cite{ref:lazyx2} can be employed to further reduce the number of function evaluations required and still provide similar near-optimality guarantees.
}
\subsection{Numerical Examples}
\begin{figure}
\centering
    \psfrag{Probability of Error}[Bc][c]{\footnotesize{Probability of Error $[P_{\rm{e}}]$}}
    \psfrag{Number of Sensors [k]}[cc][c]{\footnotesize{Number of Sensors $[K]$}}
	\psfrag{Exhaustive Search}{\fontsize{6}{6}\selectfont{Exhaustive Search}}
	\psfrag{Greedy}{\fontsize{6}{6}\selectfont{Greedy SNR}}
	\psfrag{Submodular Surrogate}{\fontsize{6}{6}\selectfont{Submodular Surrogate}}
	\psfrag{Convex Method}{\fontsize{6}{6}\selectfont{Convex Method}}
	\psfrag{Worst Subset}{\fontsize{6}{6}\selectfont{Worst Subset}}		
  \includegraphics[width=\columnwidth, height=2in]{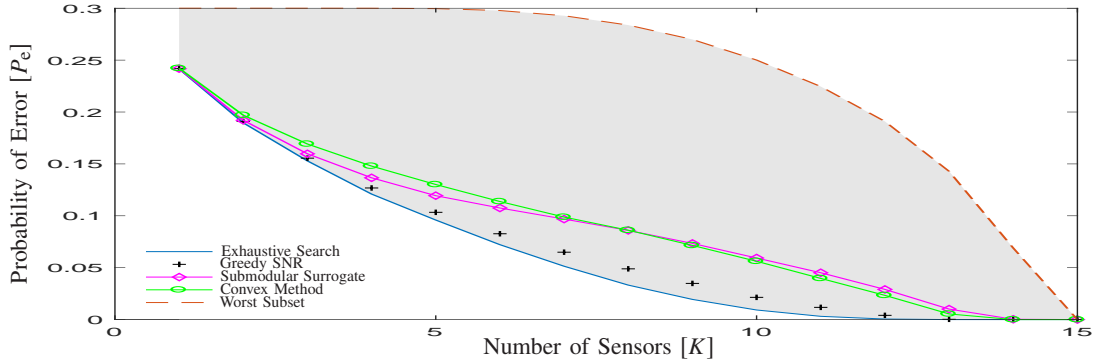}
\caption{Bayesian probability of error $P_{\rm{e}}$ for~\eqref{eq:hp2} with different subset sizes $K$ when choosing from $M=15$ available sensors. The probability of error for any random subset of $K$ sensors will be in the shaded region of the plot.}
  \label{fig:snrSmallSet}
\end{figure}
\begin{figure*}
\centering
\begin{subfigure}[t]{.3\textwidth}
  \centering
  \psfrag{Probability of Error}[Bc][c]{\scriptsize{Probability of Error $[P_{\rm{e}}]$}}
    \psfrag{Number of Sensors [k]}[cc][c]{\scriptsize{Number of Sensors $[K]$}}
	\psfrag{Greedy}{\fontsize{6}{6}\selectfont{Greedy SNR}}
	\psfrag{Submodular Surrogate}{\fontsize{6}{6}\selectfont{Submodular Surrogate}}
	\psfrag{Convex Method}{\fontsize{6}{6}\selectfont{Convex Method}}
  \includegraphics[width=.95\textwidth, height = 1.7in]{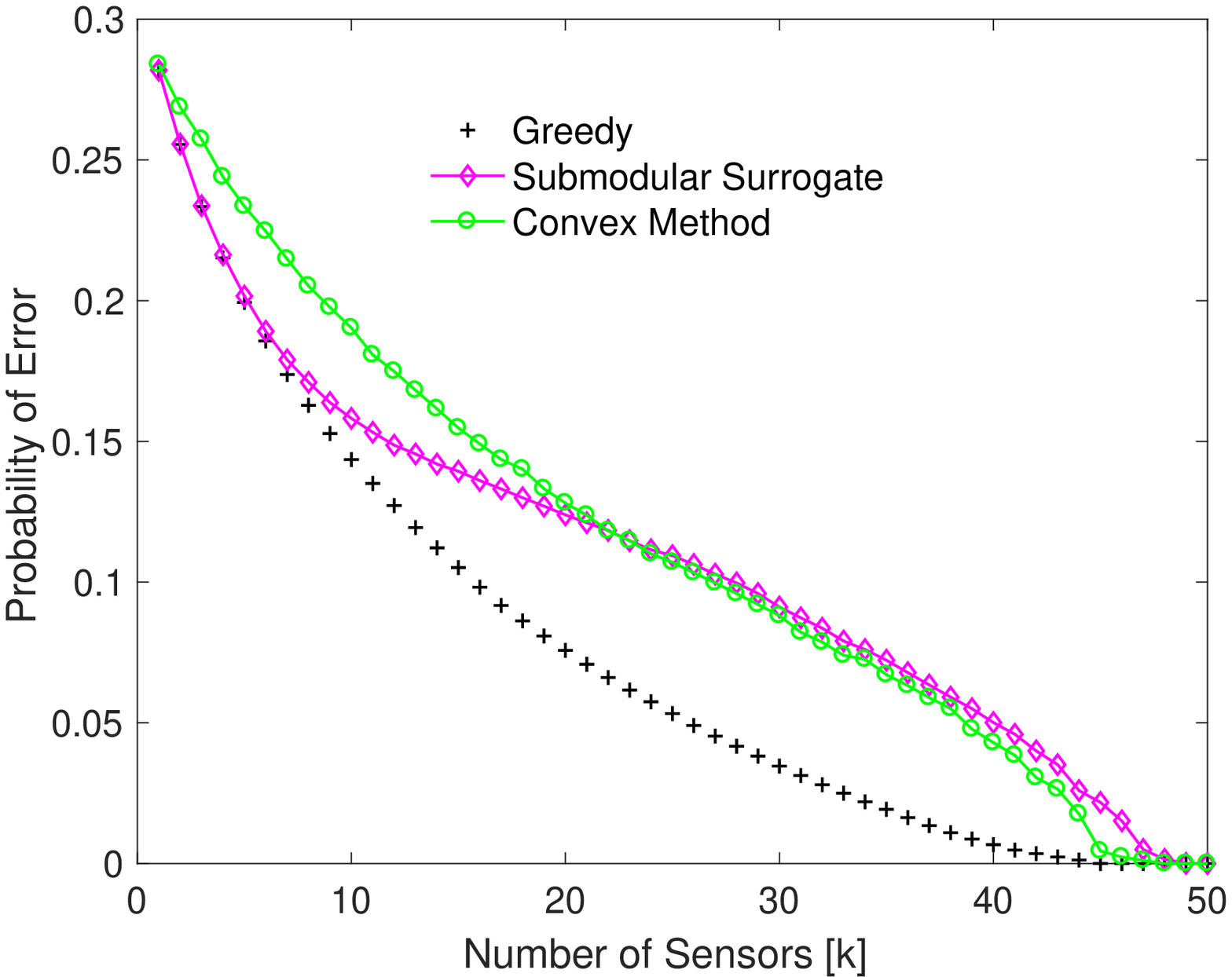}
  \caption{}
  \label{fig:snrLargeSet}
\end{subfigure}%
\begin{subfigure}[t]{0.3\textwidth}
	\centering
	\psfrag{a}[Bc][c]{\scriptsize{Percentage of Occurrence $[\%]$}}
	\psfrag{Ratio of the Solution X/CVX}[cc][c]{\scriptsize{Ratio of the solution X/Convex}}
	\psfrag{Greedy}{\fontsize{6}{6}\selectfont{Greedy SNR}}
	\psfrag{Submodular Surrogate}{\fontsize{6}{6}\selectfont{Submodular Surrogate}}
	\includegraphics[width=0.95\textwidth, height = 1.7in]{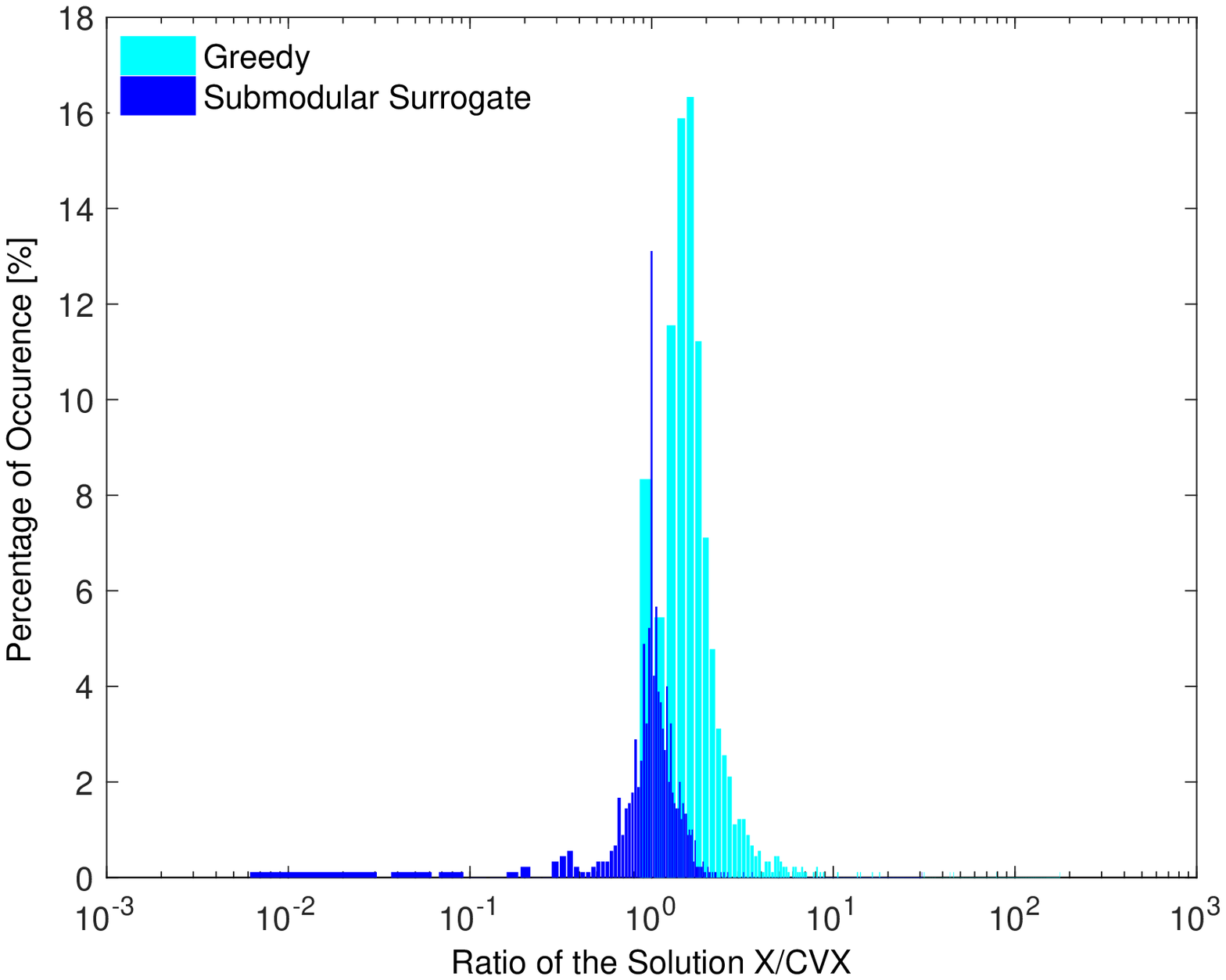}
	\caption{}
	\label{fig:snrHist}
	\end{subfigure}%
\begin{subfigure}[t]{.3\textwidth}
  \centering
  \psfrag{Probability of Error}[Bc][c]{\scriptsize{Probability of Error $[P_{\rm{e}}]$}}
    \psfrag{Number of Sensors [k]}[cc][c]{\scriptsize{Number of Sensors $[K]$}}
	\psfrag{Greedy}{\fontsize{6}{6}\selectfont{Greedy SNR}}
	\psfrag{Submodular Surrogate}{\fontsize{6}{6}\selectfont{Submodular Surrogate}}
	\psfrag{Convex Method}{\fontsize{6}{6}\selectfont{Convex Method}}
\includegraphics[width=0.95\textwidth, height = 1.7in]{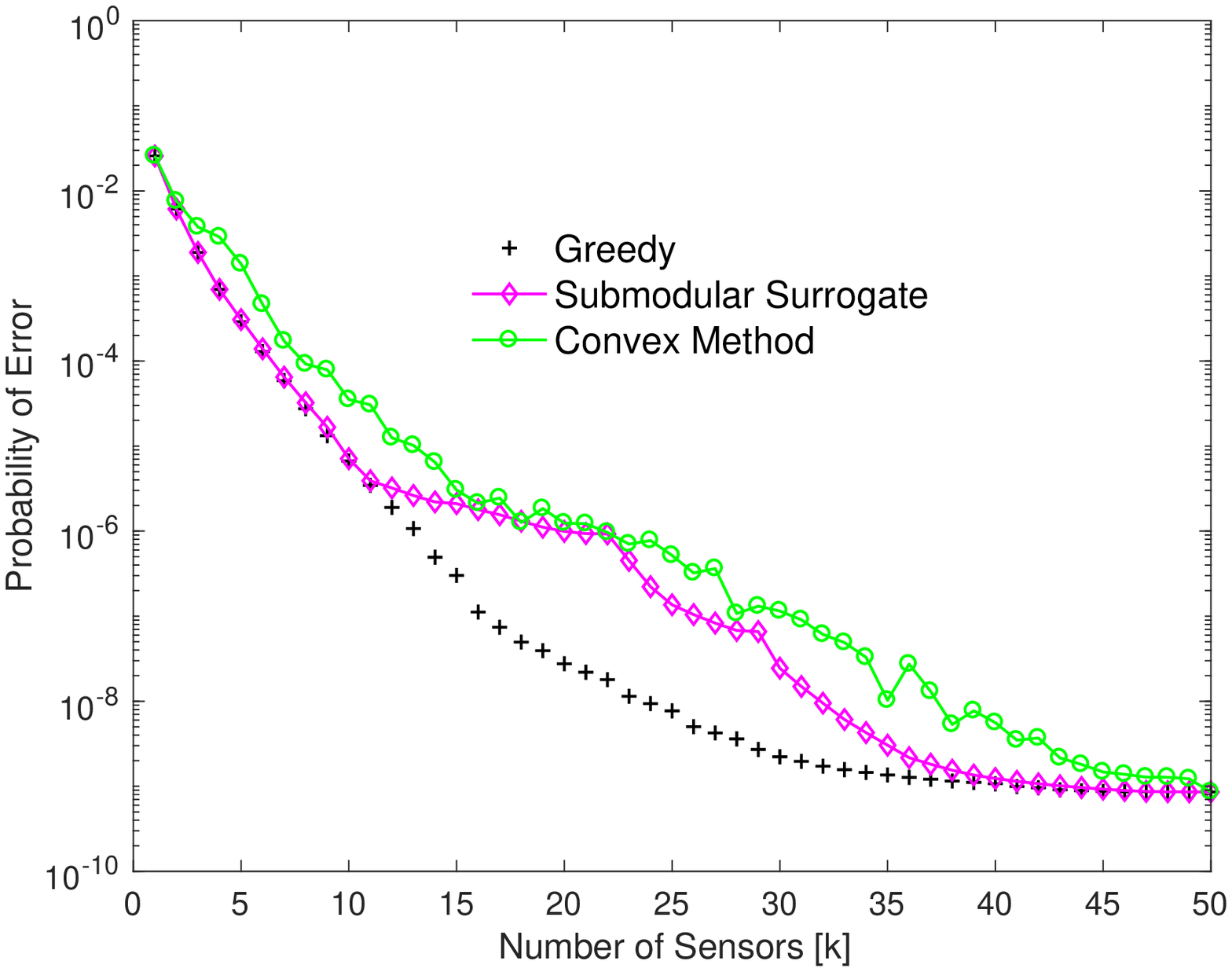}
\caption{}
\label{fig:snrLargeSetUniform}
\end{subfigure}
\caption{Results for the random Toeplitz and uniform covariance matrices. (a) Bayesian probability of error $P_{\rm{e}}$ between the convex relaxation and the greedy heuristic for~\eqref{eq:hp2} with different subset sizes $K$ when choosing from $M=50$ available sensors and random Toeplitz covariance matrices. (b) Bayesian probability of error $P_{\rm{e}}$ between the convex relaxation and the greedy heuristic for~\eqref{eq:hp2} with different subset sizes $K$ when choosing from $M=50$ available sensors and a uniform covariance matrix [c.f.~\eqref{eq:eqcov}]. (c) Histogram of the distribution of the gain in SNR of different sensor selection strategies when the relaxed convex problem is considered as baseline. The sensor selection problem is solved for $M=50$ available sensors over several realizations and different subset sizes. The height of the bar represents the relative frequency of the gain in the $x$-axis.}
\label{fig:test}
\end{figure*}
To illustrate the performance of the submodular optimization machinery, we present two different examples for~\eqref{eq:hp2} under the Bayesian setting with $P(\mathcal{H}_{0}) = 0.3$. First, let us consider a small-scale sensor selection problem where the best $K$ sensors have to be selected from a pool of $M=15$ available sensors. This small scale example allows us to compare the developed algorithm with the optimal solution. In this example, $1000$ Monte-Carlo runs are performed. The common covariance matrix $\bs{\Sigma}$, in each Monte-Carlo run, is generated using a superposition of $M$ unit power Gaussian sources in an array signal processing model~\cite{VT1},
and the mean vectors $\bs{\theta}_{i}$ are considered i.i.d. Gaussian random unit vectors. We solve the problem by performing an exhaustive search over all possible ${M}\choose{K}$ combinations. The subset that maximizes and minimizes the $P_{\rm{e}}$ of the system is obtained and represents the worst and best possible selection, respectively. In addition, a comparison between the average performance of the greedy algorithm and the convex relaxation of the problem is shown in Fig.~\ref{fig:snrSmallSet}. In the plot, the $P_{\rm{e}}$ obtained by applying directly the greedy heuristic to the signal-to-noise ratio set function is denoted as \emph{Greedy SNR}. From Fig.~\ref{fig:snrSmallSet}, it is seen that even though the submodular surrogate, given by expression~\eqref{eq:subs2}, does not perform as good as optimizing the original signal-to-noise ratio set function, its performance is comparable to the one obtained by the convex relaxation approach. However, applying Algorithm~\ref{al:greedy} to the submodular surrogate incurs a lower computational complexity due to its recursive implementation. In Fig.~\ref{fig:snrSmallSet}, the shaded area shows the region where all other sub-optimal samplers would lie for this problem. 

The previous example was intended to illustrate the performance of the discussed methods in comparison with the exhaustive search. However, for interesting problem sizes, exhaustive search solutions are not feasible even for small subset cardinalities. To illustrate the performance of the submodular surrogate for larger problem sizes, in the following example, instead of using the exhaustive search result as baseline, we compare the greedy heuristic with the convex relaxation for a problem of size $M=50$. In Fig.~\ref{fig:snrLargeSet}, the average performance over $1000$ Monte-Carlo runs is shown, when the common covariance matrix $\bs{\Sigma}$, is considered to be a random Toeplitz symmetric matrix, and the mean vector i.i.d. Gaussian as before. Similar to the results from the previous example, the greedy rule from Algorithm~\ref{al:greedy} provides the lowest $P_{\rm{e}}$ when it is applied to the original signal-to-noise ratio function. As before, the submodular surrogate provides subsets with comparable system performance as the convex relaxation method with randomization, but with a reduced computational cost.

In Fig.~\ref{fig:snrHist}, we show the ratio between the SNR of the greedy and the submodular surrogate with respect to the solution of the relaxed convex problem for $100$ Monte-Carlo realizations of problem~\eqref{eq:hp2} when random Toeplitz covariance matrices are considered for $\bs{\Sigma}$. The percentage of occurrence is shown in the vertical axis of the bar plot. In each Monte-Carlo run, the solution using the three approaches was computed, for the subset sizes $K = \{1, 6,11,16,21,26,31,36,41,46\}$. The histogram is computed over all subset sizes for each of the methods. It is evident from Fig.~\ref{fig:snrHist} that the greedy heuristic, when applied to the original signal-to-noise ratio, provides the best performance of all methods. As expected, the submodular surrogate set function provides similar results as the convex relaxation due to the fact that both are constructed from the Schur complement.

Finally, in Fig.~\ref{fig:snrLargeSetUniform} the comparison of the different methods is shown for the case when the covariance matrix $\bs \Sigma$ is considered to be a uniform correlation covariance matrix, i.e., 
\begin{equation}
\boldsymbol{\Sigma} = \begin{bmatrix}
1 & \rho & \rho & \ldots & \rho \\
\rho & 1 & \rho & \ldots &\rho \\
\vdots & \vdots & \ddots & \vdots & \vdots \\
\rho & \rho & \rho & \ldots & 1
\end{bmatrix},
\label{eq:eqcov}
\end{equation}
with correlation factor $\rho = 0.43$. From Fig.~\ref{fig:snrLargeSetUniform} it is seen that the submodular surrogate outperforms the convex relaxation for all subset sizes. However, the best performance is achieved by the greedy heuristic applied directly to the signal-to-noise ratio set function.



\subsection{When Does Greedy on the SNR Fail?}
In the previous part, it has been numerically shown that the greedy heuristic applied directly to the signal-to-noise ratio set function might perform better than both the convex and submodular relaxations of the problem. However, we should be aware that the application of the greedy heuristic for a non-submodular maximization does not provide any optimality guarantees in general. Therefore, there might be problem instances in which the direct maximization of such a set function could lead to arbitrary bad results. In order to illustrate the importance of submodularity for the greedy heuristic, we show an example of the sensor selection problem in which the greedy method applied  to the signal-to-noise performs worse than the submodular surrogate. Consider an example with $M=3$ available sensors, from which we desire to obtain the best subset of $K=2$ sensors which provides the highest signal-to-noise ratio. In addition, we consider the case where the difference of the mean vectors is the all-one vector, i.e., $\boldsymbol{\theta} = [1,1,1]^{T}.$ The covariance matrix for the noise is given by the block matrix
$$
\boldsymbol{\Sigma} = \begin{bmatrix}
1/(1-\rho^{2}) & -\rho/(1-\rho^{2}) & 0 \\
-\rho/(1-\rho^{2}) & 1/(1-\rho^{2}) & 0 \\
0 & 0 & 1
\end{bmatrix},
$$
where $\rho\in[0,1)$. The signal-to-noise ratio set function is defined as
$
s(\mathcal{A}) = \textbf{1}^{T}_{\mathcal{A}}\boldsymbol{\Sigma}^{-1}_{\sset{A}}\textbf{1}_\mathcal{A}.
$
Since $s({\{1\}}) = s({\{2\}}) = (1-\rho^{2})$ and $s({\{3\}})=1$, Algorithm~\ref{al:greedy} will select $\{3\}$ first as $\rho \leq 1$, i.e., $\mathcal{A}_{1} = \{3\}$. Then, either $\{1\}$ or $\{2\}$ are chosen next as both have the same set function value, i.e.,
$$
s({\{3,1\}}) = s({\{3,2\}}) = s(\mathcal{A}_{\sset{G}}) =2 -\rho^{2},
$$
where $\sset{A}_{\sset{G}}$ denotes the set obtained from the greedy SNR solution, i.e., obtained by greedily maximizing the SNR. However, the maximum of the set function is attained with the set $\mathcal{A}^{*} = \{1,2\}$ which provides the set function value $s(\textbf{w}_{\mathcal{A}^{*}}) = 2 + 2\rho$. For the limiting case $\rho\rightarrow 1$, we obtain
$$
\underset{\rho\rightarrow 1}{\lim}\;\frac{s({\mathcal{\mathcal{A}_{G}}})}{s({\mathcal{A}^{*}})} = 0.25.
$$
Even though the greedy heuristic can provide good results in many cases, one should thus be aware that it could get stuck in solutions far from the optimal.

We will now show for the above example that, on average, applying the greedy heuristic to the submodular surrogate performs better than applying it to the original SNR cost set function. First, let us consider the following decomposition of the noise covariance matrix,
\begin{equation}
\textbf{S} = \bs\Sigma - a\textbf{I} =\begin{bmatrix}
\frac{1}{1 - \rho^2} - a & -\frac{\rho}{1 - \rho^2} & 0\\ -\frac{\rho}{1 - \rho^2} &  \frac{1}{1 - \rho^2} - a & 0\\ 0 & 0 & 1 - a
\end{bmatrix},
\label{eq:Sexam}
\end{equation}
for any $a$ chosen as described in Theorem~\ref{th:epSub}.

Then, the inverse of~\eqref{eq:Sexam} can be expressed as
$$
\textbf{S}^{-1} = \begin{bmatrix}
\frac{a\, \rho^2 - a + 1}{a^2\, \rho^2 - a^2 + 2\, a - 1} & -\frac{\rho}{a^2\, \rho^2 - a^2 + 2\, a - 1} & 0\\ -\frac{\rho}{a^2\, \rho^2 - a^2 + 2\, a - 1} & -\frac{a\, \rho^2 - a + 1}{a^2\, \rho^2 - a^2 + 2\, a - 1} & 0\\ 0 & 0 & -\frac{1}{a - 1} 
\end{bmatrix}.
$$
The submodular cost set function can be evaluated for each of the sensors by considering its factors as in~\eqref{eq:prod}, i.e.,
$$
\begin{array}{rl}
\gamma({\{i\}}) & =\det\big(\textbf{S}^{-1} + a^{-1}\textbf{I}_{{\{i\}}}\big) = \frac{1}{a\, \left(a - 1\right)\, \left(a^2\, \rho^2 - a^2 + 2\, a - 1\right)},\\
s({\{i\}}) &= 1 - \rho^2,\;\text{for }i=1,2,
\end{array}
$$
and
$$
\begin{array}{rl}
\gamma({\{3\}}) & =\det\big(\textbf{S}^{-1} + a^{-1}\textbf{I}_{{\{3\}}}\big) = \frac{1 - \rho^2}{a\, \left(a - 1\right)\, \left(a^2\, \rho^2 - a^2 + 2\, a - 1\right)}\\
s({\{3\}}) &= 1.
\end{array}
$$
It is clear that the submodular cost set function provides the same value for any of the sensors, i.e.,
$$
\gamma({\{1\}})s({\{1\}}) = \gamma({\{2\}})s({\{2\}}) = \gamma({\{3\}})s({\{3\}}).
$$
Hence, if we break this tie arbitrarily, the possible values of the cost set function are
$$
\begin{array}{ll}
\gamma({\{1,2\}})s({\{1,2\}}) = \gamma({\{2,1\}})s({\{2,1\}}) =   \frac{2 + 2\rho}{a^2\, \left(a - 1\right)\, \left(a^2\, \rho^2 - a^2 + 2\, a - 1\right)}\\
\gamma({\{3,1\}})s({\{3,1\}}) = \gamma({\{3,2\}})s({\{3,2\}}) =  \frac{2 - \rho^2}{a^2\, \left(a - 1\right)\, \left(a^2\, \rho^2 - a^2 + 2\, a - 1\right)},
\end{array}
$$
where we consider the fact that the greedy heuristic does not select the $3$rd sensor after the $1$st or the $2$nd sensor has been selected, i.e., the marginal gain is larger when the sensors $\{1,2\}$ are selected. Therefore, the average value attained by the submodular method is
$$
\begin{array}{ll}
E\big[s({\mathcal{A}_{\mathcal{S}}})\big] &= \frac{1}{3}(s({\{1,2\}}) + s({\{2,1\}})) + \frac{1}{6}(s({\{3,1\}}) + s({\{3,2\}})) \\&= \frac{2}{3}s({\mathcal{A}^{*}}) + \frac{2-\rho^{2}}{3},
\end{array}
$$
where $\mathcal{A}_{\mathcal{S}}$ is the set returned by the maximization of the submodular surrogate. In the limiting case $\rho\rightarrow 1$, we have the following limit
$$
\underset{\rho\rightarrow 1}{\lim}\; \frac{E\big[s({\mathcal{A}_{\mathcal{S}}})\big]}{s({\mathcal{A}^{*}})}  = 0.75,
$$
which provides a higher approximation ratio compared with the previously seen greedy heuristic. However, it is clear that the proposed method also suffers from one of the drawbacks of greedy methods: when more than one possible solution obtains the same cost set function value, either ties should be broken arbitrarily or multiple branches have to be initialized.

Now, we show a larger instance of the previous example, where a set of $M=200$ available sensors are considered. Furthermore, a block precision matrix $\bs \Sigma^{-1}$ with the following structure is considered for performing sensor selection
\begin{equation}
\bs \Sigma^{-1} = \begin{bmatrix}
\bm T & \bm 0 \\
\bm 0 & \bm I
\end{bmatrix}\in\mathbb{R}^{M\times M}
\label{eq:premtx}
\end{equation}
where $\bm T = \rm{Toeplitz}([1,\rho^1,\rho^2,\ldots,\rho^{M-1}])\in\mathbb{R}^{M/2\times M/2}$ is an exponential decaying Toeplitz matrix, and $\bm I\in\mathbb{R}^{\lfloor M/2\rfloor \times \lfloor M/2 \rfloor}$ is the identity matrix. This kind of precision matrices could arise in systems where only a subset of sensors are calibrated, i.e., block of sensors whose precision matrix is the identity. The mean difference vector, i.e., $\boldsymbol\theta = \boldsymbol{\theta}_{1} -\boldsymbol{\theta}_0$, is considered the all-ones vector, and ties in the selection are broken arbitrarily. In this example $\rho = 0.18$ has been fixed. From Fig.~\ref{fig:snrLargeSetUniform2} it can be seen that even though for a small number of selected sensors both methods achieve similar SNR, the submodular surrogate outperforms the Greedy SNR method for most of the subset sizes. This result is expected due to the fact that the worst case bound given in Theorem~\ref{th:epSub} for $\epsilon$-submodular set functions worsen as the size of the solution increases. More importantly, the submodular surrogate reaches the maximum SNR when half the sensors have been selected, whereas the Greedy SNR requires all the sensors to reach the maximum SNR.
 
\begin{figure}
\centering
    \psfrag{SNR [dB]}[Bc][c]{\footnotesize{SNR [dB]}}
    \psfrag{Number of Sensors [k]}[Bc][c]{\footnotesize{Number of Sensors $[K]$}}
	\psfrag{Greedy SNR}{\fontsize{6}{6}\selectfont{Greedy SNR}}
	\psfrag{Submodular Surrogate}{\fontsize{6}{6}\selectfont{Submodular Surrogate}}
\includegraphics[width=0.5\textwidth, height = 2in]{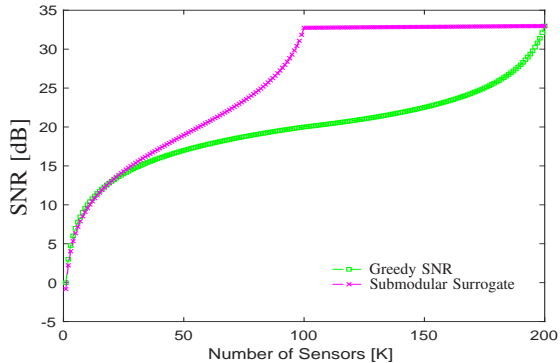}
\caption{Signal-to-noise ratio  between the Greedy SNR and the submodular surrogate for different subset sizes $K$ when choosing from $M=200$ available sensors for an instance of the problem with precision matrix given by~\eqref{eq:premtx}.}
\label{fig:snrLargeSetUniform2}
\end{figure}

\section{Observations with Uncommon Covariances}
\label{sec:ssuc}
In this section, we discuss sensor selection for detection when the data model for the hypotheses under test differ in their second-order statistics. For the case of Gaussian distributed measurements we can assume that the conditional distributions for the binary hypothesis test are given by
\begin{equation}
\begin{array}{c}
\sset{H}_{0} : \bm{x}\sim\sset{N}(\bs{\theta},\bs{\Sigma}_{0})\\
\sset{H}_{1} : \bm{x}\sim\sset{N}(\bs{\theta},\bs{\Sigma}_{1}),
\end{array}
\label{eq:hp3}
\end{equation}
where the mean vector $\bs{\theta}\in\mathbb{R}^{M}$ is shared by both hypotheses and the second-order statistics of the data are characterized by the $M\times M$ covariance matrices $\bs{\Sigma}_{0}$ and $\bs{\Sigma}_{1}$ for the hypothesis $\sset{H}_{0}$ and $\sset{H}_{1}$, respectively.

Using the metrics discussed in Section~\ref{sec:2}, it has been suggested in~\cite{sub:r15} that the different distance measures between probability distributions can be applied to construct selection strategies using convex optimization. However, it turns out that some of the metrics to optimize can only be expressed as the difference of submodular functions, therefore the SupSub procedure described in Section~\ref{sec:pre} can be employed for its optimization. In the next section, we show how it is possible to decompose the divergence measures into the difference of submodular functions.

\subsection{Submodular Decomposition of Divergence Measures}
Unlike the case with commons means, the three distances discussed are not scaled versions of each other. For the linear model in~\eqref{eq:linModelY} under Gaussian noise, the Bhattacharyya distance~\eqref{eq:bhatDiv1} is given as the following difference of submodular set functions
\begin{align}
f({\sset{A}}) &= \mathcal{B}(\mathcal{H}_{1}\Vert\mathcal{H}_{0})\coloneq g({\mathcal{A}})-h({\mathcal{A}}); \\
g(\mathcal{A}) &= \frac{1}{2}\log\det(\bs{\Sigma}_{\sset{A}}); \nonumber \\
h(\mathcal{A}) &= \frac{1}{4}(\log\det(\bs{\Sigma}_{0,\sset{A}})+\log\det(\bs{\Sigma}_{1,\sset{A}})) \nonumber.
\end{align}

The submodularity of $h({\mathcal{A}})$ and $g({\mathcal{A}})$ is clear as both functions are linear combinations of entropy functions. As a result, the Bhattacharyya distance can be approximately maximized using the SupSub procedure described in Algorithm~\ref{al:supsub}.  

Differently from the Bhattacharyya distance, the expressions for the KL divergence and the J-divergence in~\eqref{eq:kldivdi} and~\eqref{eq:jedi} for the distributions in~\eqref{eq:hp3} do not provide a direct decomposition in submodular set functions because in both divergences there are trace terms that cannot be expressed directly as a difference of submodular functions. Even though such decompositions exist~\cite{sub:r11}, in general, finding them incurs exponential complexity~\cite{sub:r12}. However, similarly as in the case of the signal-to-noise ratio cost set function, a readily available submodular surrogate can be employed in order to optimize both distances using the SupSub procedure.

In order to obtain a submodular approximation of the trace term, let us consider the following set function
\begin{equation}
q(\boldsymbol\Sigma_{\mathcal{A}},\boldsymbol\Psi_{\mathcal{A}}) = \text{tr}\{\boldsymbol\Sigma^{-1}_\sset{A}\boldsymbol\Psi_\sset{A}\},
\label{eq:psiset}
\end{equation}
where $\mathcal{A}$ is the index set of the selected sensors and $\boldsymbol{\Sigma}_{\sset{A}}$ and $\boldsymbol{\Psi}_\sset{A}$ are submatrices defined by the rows and columns of $\boldsymbol{\Sigma}$ and $\boldsymbol{\Psi}$, respectively, given by the elements of the set $\sset{A}$. Let us decompose one of the matrices as
$
\boldsymbol\Sigma = a\textbf{I} + \textbf{S},
$
where a nonzero $a\in\mathbb{R}$ is chosen as described in Theorem~\ref{th:epSub} and therefore $\textbf{S}\succ 0$. The set function in~\eqref{eq:psiset} is then equivalent to
\begin{equation*}
\begin{aligned}
q(\boldsymbol\Sigma_{\mathcal{A}},\boldsymbol\Psi_{\mathcal{A}}) &= \text{tr}\{\textbf{S}^{-1}\boldsymbol\Psi - \textbf{S}^{-1}\big[\textbf{S}^{-1} + a^{-1}\text{diag}(\textbf{1}_{\mathcal{A}})\big]^{-1}\textbf{S}^{-1}\boldsymbol\Psi
\}\\
&= \text{tr}\{\boldsymbol\Psi^{\frac{T}{2}}\textbf{S}^{-\frac{1}{2}}\big(
\textbf{I} - \textbf{S}^{-\frac{T}{2}}[\textbf{S}^{-1}\\
& \qquad + a^{-1}\text{diag}(\bm{1}_\mathcal{A})]^{-1}\textbf{S}^{-\frac{1}{2}}
\big)\textbf{S}^{-\frac{T}{2}}\boldsymbol\Psi^{\frac{1}{2}}
\}\\
&= \sum\limits_{i=0}^{M}\text{tr}\{\bm z_{i}^{T}\big( 
\textbf{I} - \textbf{S}^{-\frac{T}{2}}[\textbf{S}^{-1}\\
& \qquad+ a^{-1}\text{diag}(\bm{1}_{\mathcal{A}})]^{-1}\textbf{S}^{-\frac{1}{2}}
\big)\bm z_{i} \},
\label{eq:subpsi}
\end{aligned}
\end{equation*}
where $\bm z_{i}$ has been defined as the $i$-th column of $\textbf{S}^{-\frac{T}{2}}\boldsymbol\Psi^{\frac{1}{2}}$. Analogously to the uncommon means case, where the signal-to-noise ratio was replaced by its submodular surrogate, we can substitute $q(\boldsymbol\Sigma_{\mathcal{A}},\boldsymbol\Psi_{\mathcal{A}})$ by the following submodular set function
\begin{equation*}
q_{\rm{sub}}(\boldsymbol\Sigma_{\mathcal{A}},\boldsymbol\Psi_{\mathcal{A}}) \coloneq \sum\limits_{i=1}^{M}\log\det\begin{bmatrix}
\textbf{S}^{-1}+a^{-1}\text{diag}(\textbf{1}_{\mathcal{A}}) & \textbf{S}^{-\frac{1}{2}}\bm z_{i}\\
\bm z_{i}^{T}\textbf{S}^{-\frac{T}{2}} & \bm z_{i}^{T}\bm z_{i}
\end{bmatrix}
\label{eq:psifin}
\end{equation*}
which is submodular on the set of selected entries ${\mathcal{A}}$. It is clear that the set function $q_{\rm{sub}}(\boldsymbol\Sigma_{\mathcal{A}},\boldsymbol\Psi_{\mathcal{A}})$ is submodular as it is a non-negative combination of submodular set functions in~$\mathcal{A}$. Furthermore, as this set function shares a similar structure with respect to the signal-to-noise ratio set function [cf.~\eqref{eq:csnr}], i.e.,
\begin{equation}
\begin{array}{ll}
q_{\rm{sub}}(\boldsymbol\Sigma_{\sset{A}},\boldsymbol\Psi_{\sset{A}}) &\coloneq M\log\det(\bm S^{-1}+a^{-1}\text{diag}(\bm 1_{\sset{A}}))+\\ &\quad\quad\sum\limits_{i=1}^{M}\log(\boldsymbol\psi_{i}^T\boldsymbol\Phi_{\sset{A}}^T\boldsymbol{\Sigma}_{\sset{A}}^{-1}\boldsymbol{\Phi}_{\sset{A}}\boldsymbol{\psi}_{i}),
\label{eq:psifin2}
\end{array}
\end{equation}
where $\boldsymbol\psi_{i}$ is the $i$th column of $\boldsymbol{\Psi}^{\frac{1}{2}}$, an efficient evaluation of~\eqref{eq:psifin2} can be performed through a recursive definition similar to the one in~\eqref{eq:recCost}. Unfortunately, as the summation is over $M$ terms, this formulation leads to a worst-case complexity of $\mathcal{O}(M^2K^3)$ for finding the solution through a greedy heuristic. However, for instances with $K \ll M$ this algorithm improves, in terms of speed, with respect to the cubic complexity of the convex relaxation.

After the introduction of the submodular set function $q_{\rm{sub}}$, surrogates for the divergences $\mathcal{K}(\cdot)$ and $\mathcal{D}_{\rm J}(\cdot)$ denoted as $\mathcal{K}_{\rm{sub}}(\cdot)$ and $\mathcal{D}_{\rm{J,sub}}(\cdot)$, respectively, can be obtained. The following is observed from these surrogates:
\begin{itemize}
\item $\mathcal{K}_{\rm{sub}}(\cdot)$ can be expressed as a mixture of submodular and supermodular set functions as
\begin{align*}
&\mathcal{K}_{\rm{sub}}(\mathcal{H}_{1}\Vert\mathcal{H}_{0}) = g(\mathcal{A}) - {h}(\mathcal{A});\\
g(\mathcal{A}) &= \frac{1}{2}\log\det(\bs{\Sigma}_{0,\sset{A}}) + \frac{1}{2}q_{\rm{sub}}(\bs{\Sigma}_{0,\sset{A}},\bs{\Sigma}_{1,\sset{A}});\\
{h}(\mathcal{A}) &= \frac{1}{2}\log\det(\bs{\Sigma}_{1,\sset{A}}).
\end{align*}
\item $\mathcal{D}_{\rm{J,sub}}(\cdot)$ is a submodular set function as it is a non-negative combination of two submodular functions, i.e.,
$$
\mathcal{D}_{\rm{J,sub}}(\mathcal{H}_{1}\Vert\mathcal{H}_{0}) = \frac{1}{2}(q_{\rm{sub}}(\boldsymbol\Sigma_{0,\sset{A}},\boldsymbol\Sigma_{1,\sset{A}}) + q_{\rm{sub}}({\boldsymbol\Sigma_{1,\sset{A}},\boldsymbol\Sigma_{0,\sset{A}}})).
$$
\end{itemize}
From these results, it is clear that $\mathcal{K}_{\rm{sub}}(\cdot)$ can be optimized using the SupSub procedure in Algorithm~\ref{al:supsub} as in the case of the Bhattacharyya distance, while $\mathcal{D}_{\rm{J,sub}}(\cdot)$ can be directly optimized using the greedy heuristic from Algorithm~\ref{al:greedy}.

\subsection{Numerical Examples}
We demonstrate the applicability of the SupSub in Algorithm~\ref{al:supsub} for solving the maximization of the different divergences used for sensor selection, and its respective surrogates by comparing the results with the widely used CCP heuristic. To do so, first we perform an exhaustive search to solve the sensor selection problem for the test in~\eqref{eq:hp3} under the Neyman-Pearson setting. We find the subset of size $K$ that maximizes the KL divergence, for random covariance matrices of size $M=15$ and for random Toeplitz matrices of size $M=50$. The results are shown in Fig.~\ref{fig:covSmallSet} and Fig.~\ref{fig:covLargeSet}, respectively. From these examples, it is seen that the greedy heuristic of Algorithm~\ref{al:greedy} applied to the KL divergence (\emph{KL Greedy}), the SupSub procedure using both the original KL expression  (\emph{SupSub KL-Div})\footnote{This is done by computing the expressions of the modular upper bounds [cf.~\eqref{eq:modBndA} and~\eqref{eq:modBndB}] for the set function $q(\bs \Sigma_{\sset{A}},\bs\Psi_{\sset{A}})$, despite that the function is not submodular.} and the submodular surrogate (\emph{SupSub Surrogate}) perform either better or as good as the CCP heuristic while incurring a much lower complexity. For random Toeplitz matrices, as seen in Fig.~\ref{fig:covLargeSet}, all the methods perform close to each other, however the CCP method achieves this performance with a higher computational load. 
\begin{figure*}
\centering
	\begin{subfigure}[t]{0.5\textwidth}
	\centering
	\psfrag{Number of Sensors [k]}[tc][c]{\footnotesize{Number of Sensors $[K]$}}
	\psfrag{KL-Divergence}[Bc][c]{\footnotesize{KL-Divergence}}
	\psfrag{Exhaustive Search}{\fontsize{6}{6}\selectfont{Exhaustive Search}}
	\psfrag{SupSub Procedure}{\fontsize{6}{6}\selectfont{SupSub KL-Div}}
	\psfrag{KL-Greedy}{\fontsize{6}{6}\selectfont{KL Greedy}}
	\psfrag{SupSub (Surrogate)}{\fontsize{6}{6}\selectfont{SupSub Surrogate}}
	\psfrag{CCP Pocedure}{\fontsize{6}{6}\selectfont{CCP Procedure}}
	\includegraphics[width=0.8\textwidth, height=1.75in]{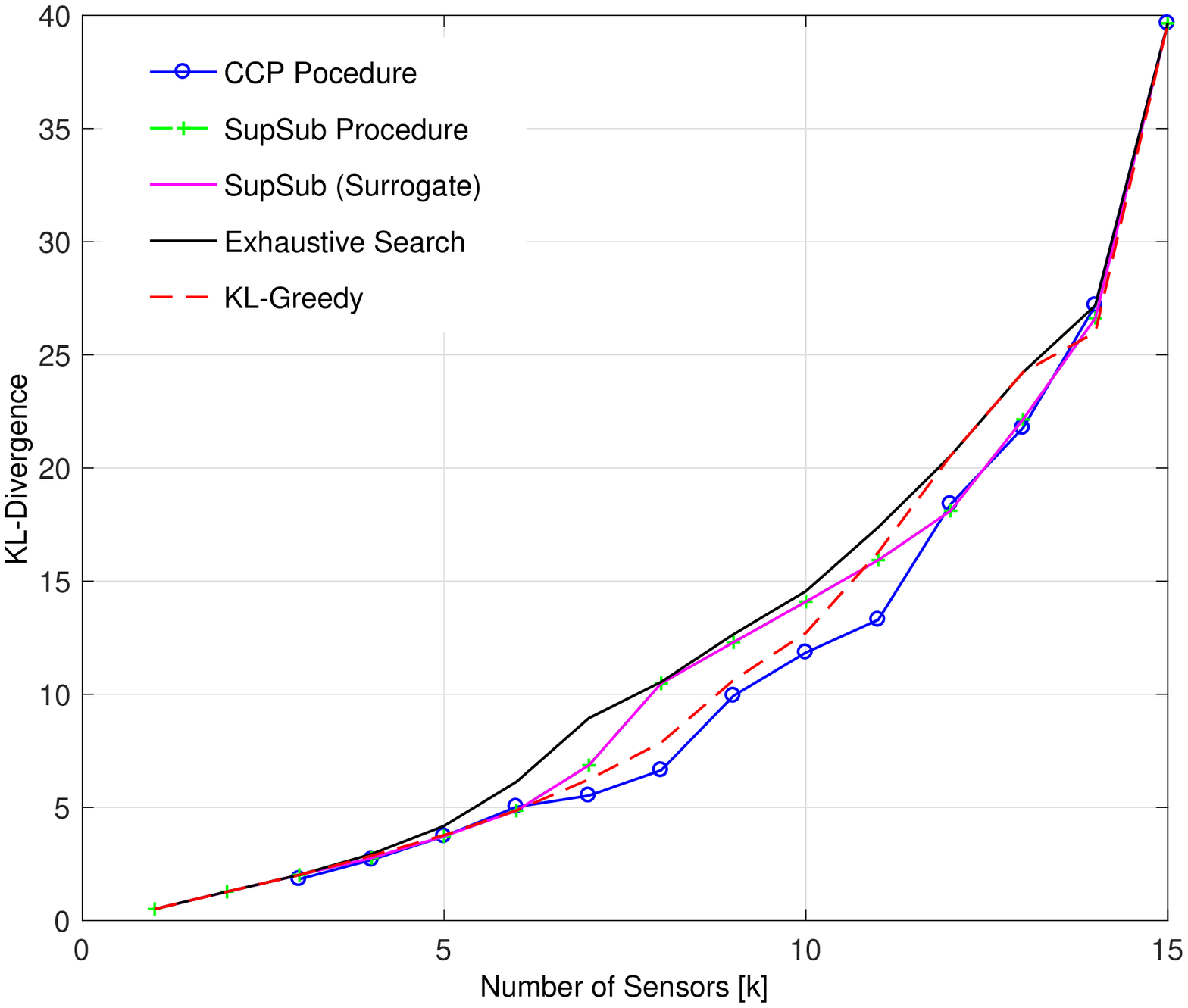}
	\caption{}
	\label{fig:covSmallSet}
	\end{subfigure}%
	\begin{subfigure}[t]{0.5\textwidth}
	\centering
	\psfrag{Number of Sensors [k]}[tc][c]{\footnotesize{Number of Sensors $[K]$}}
	\psfrag{KL-Divergence}[Bc][c]{\footnotesize{KL-Divergence}}
	\psfrag{SupSub Procedure}{\fontsize{6}{6}\selectfont{SupSub KL-Div}}
	\psfrag{KL-Greedy}{\fontsize{6}{6}\selectfont{KL Greedy}}
	\psfrag{SupSub (Surrogate)}{\fontsize{6}{6}\selectfont{SupSub Surrogate}}
	\psfrag{CCP Pocedure}{\fontsize{6}{6}\selectfont{CCP Procedure}}
	\includegraphics[width=0.8\textwidth, height=1.75in]{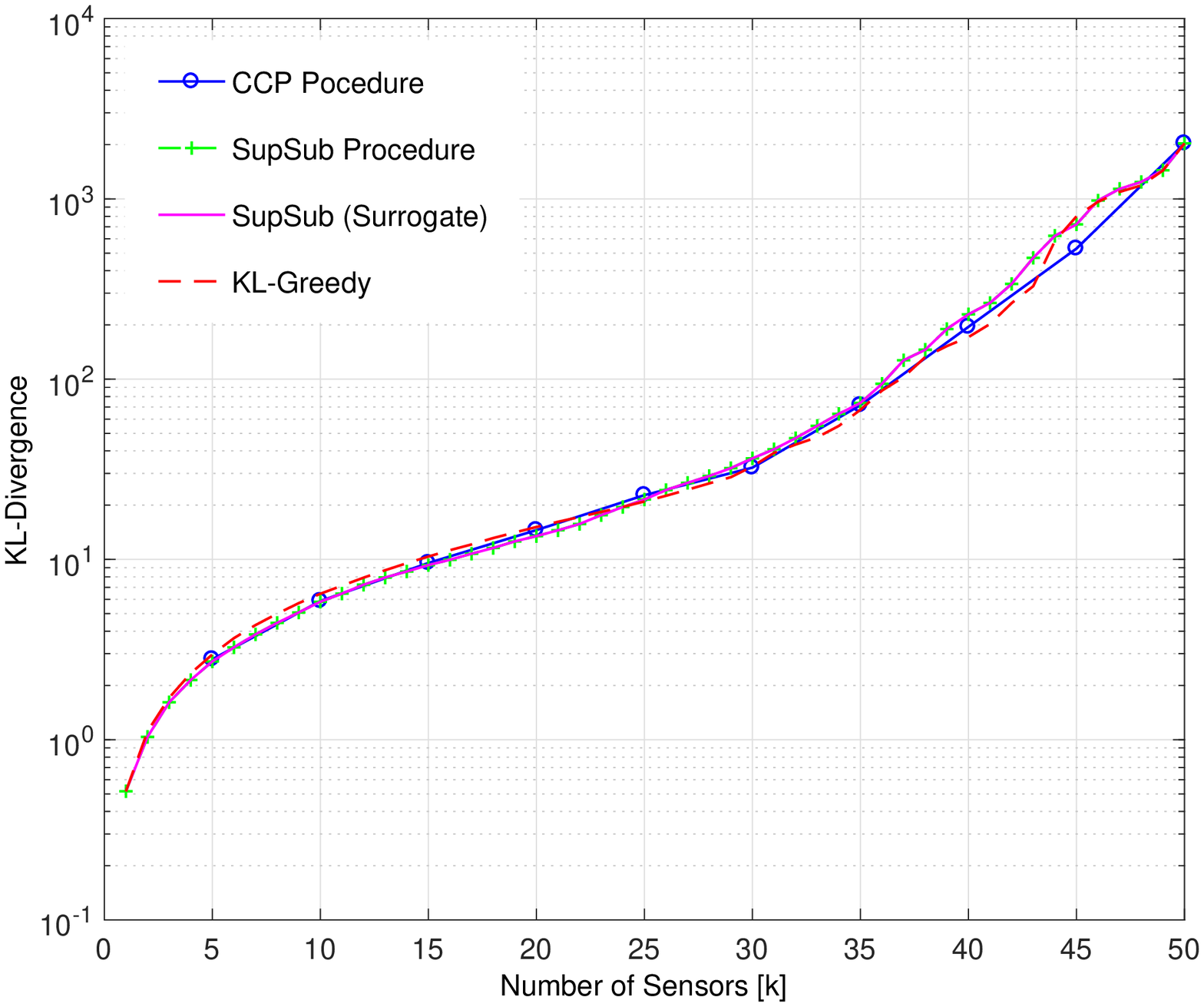}
	\caption{}
	\label{fig:covLargeSet}
	\end{subfigure}
	\caption{(a) KL divergence of the different sensor selection methods for different subset sizes $K$ for random covariance matrices. (b) KL divergence of the different sensor selection methods for different subset sizes $K$ for random Toeplitz matrices.}
\end{figure*}

\subsection*{Binary Classification}

Due to the non-monotonic behavior of the classification curves with respect to the number of features, i.e., the error of a classifier does not necessarily reduce when more features are used, a fast and reliable way to select the most relevant features for a given dataset  is required. Therefore, in the following, we present two examples for binary classification where the KL divergence is used as a feature selection metric and it is optimized using the methods described in this work. In these examples, the PRTools Toolbox~\cite{prtools} is used for training classifiers. The built-in feature selection method, based on cross-validation, is used as baseline for comparison with the proposed methods based on the submodular machinery. 

In the first example, we start by considering a simple case: two classes described by Gaussian distributions parametrized by their covariance matrices, $\{\bs\Sigma_{0}, \bs \Sigma_{1}\}$. In this scenario, the covariance matrices are a pair of Toeplitz matrices. The number of features considered for this example is $50$. The trained classifier is the quadratic discriminant classifier (QDC)~\cite{PRqdc}. Furthermore, the $20/80$ rule for training and testing has been used for the $500$ objects contained in the dataset, i.e., $20\%$ of the data set has been used for training the classifier and $80\%$ for reporting its performance on unseen data. Additionally, random sampling of the objects for training has been performed. The selection of such a classifier is due to the nature of the dataset, i.e., as the assumption of Gaussianity of the features holds, QDC is the Bayes detector for equiprobable classes. The comparison of the classification soft error for the selected classifier is shown in Fig.~\ref{fig:classQDC}. The reported error in this figure is given by
\begin{equation}
e := \frac{E_{0}}{|\mathcal{C}_{0}|}P(\mathcal{C}_{0}) + \frac{E_{1}}{|\mathcal{C}_{1}|}P(\mathcal{C}_{1}),
\end{equation}
where $E_{i}$ denotes the number of erroneously classified objects for the $i$th class, denoted by $\mathcal{C}_i$, and $P(\mathcal{C}_{i})$ represents the prior probability for the $i$th class in the validation
set.

As expected, the classification error decreases as the number of selected features increases as in this example QDC provides decision boundaries based on the log-likehood ratio test under Gaussian assumptions. In this example, both methods based on the SupSub procedure provide a similar classification error, being mostly below the PRTools baseline result. In this scenario, for roughly half the number of sensors, the greedy heuristic over the KL divergence provides the lowest classification error.

\begin{figure*}
 	\centering
	\begin{subfigure}[t]{0.5\textwidth}
	\centering
	\psfrag{Number of Features [K]}[tc][c]{\scriptsize{Number of Features $[K]$}}
	\psfrag{Classification Error}[Bc][c]{\scriptsize{Classification Error $[\%]$}}
	\psfrag{PRTools}{\fontsize{6}{6}\selectfont{PRTools}}
	\psfrag{SupSub Procedure}{\fontsize{6}{6}\selectfont{SupSub KL-Div}}
	\psfrag{KL Greedy}{\fontsize{6}{6}\selectfont{KL Greedy}}
	\psfrag{SupSub (Surrogate)}{\fontsize{6}{6}\selectfont{SupSub Surrogate}}
	\includegraphics[width=0.8\textwidth, height=1.75in]{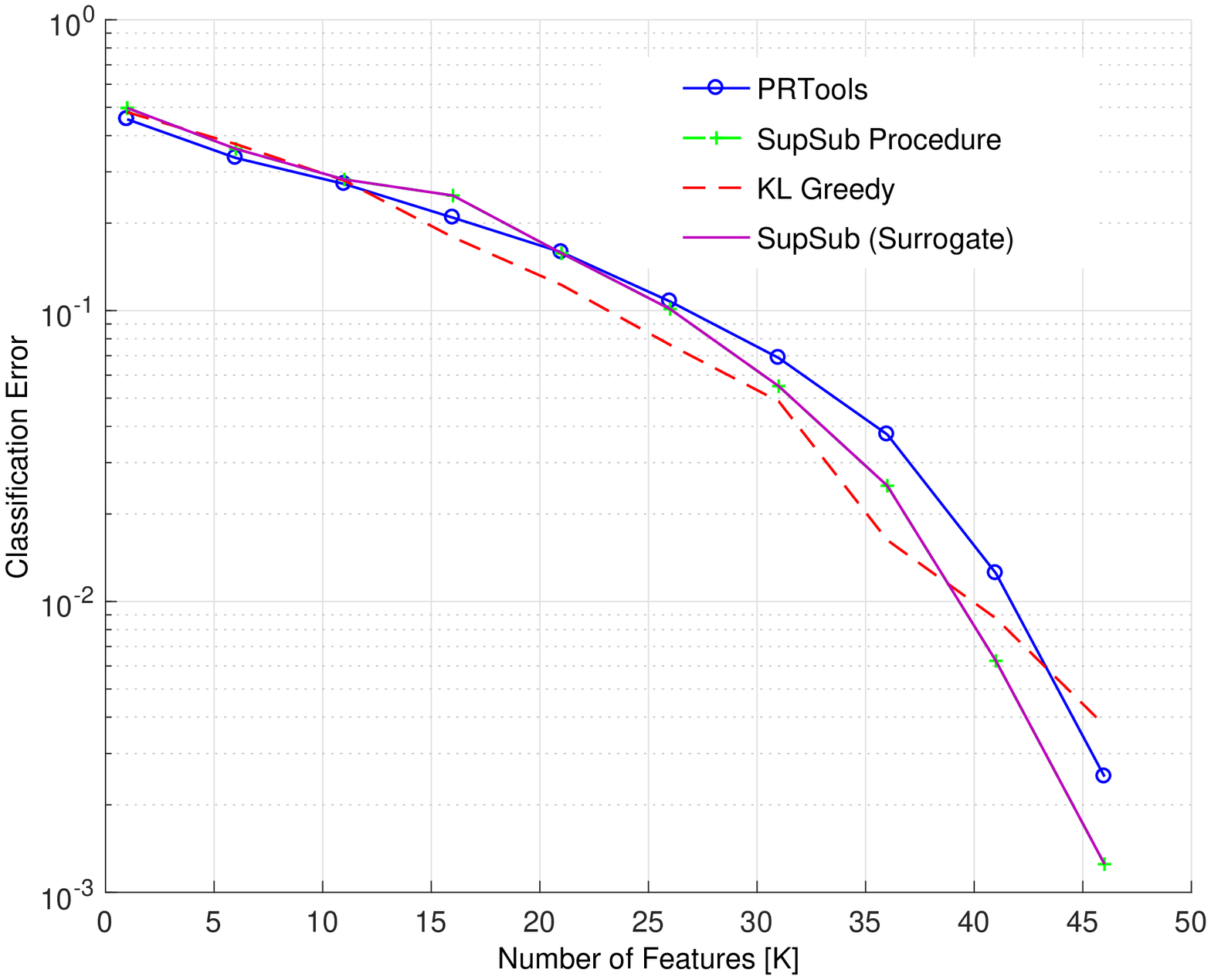}
	\caption{}
	\label{fig:classQDC}
	\end{subfigure}%
	\begin{subfigure}[t]{0.5\textwidth}
	\centering
	\psfrag{Number of Features [K]}[cc][c]{\footnotesize{Number of Features $[K]$}}
	\psfrag{Classification Error [\%]}[Bc][c]{\footnotesize{Classification Error $[\%]$}}
	\psfrag{PRTools}{\fontsize{6}{6}\selectfont{PRTools}}
	\psfrag{SupSub Procedure}{\fontsize{6}{6}\selectfont{SupSub KL-Div}}
	\psfrag{KL Greedy}{\fontsize{6}{6}\selectfont{KL Greedy}}
	\psfrag{SupSub (Surrogate)}{\fontsize{6}{6}\selectfont{SupSub Surrogate}}
	\includegraphics[width=0.8\textwidth, height=1.75in]{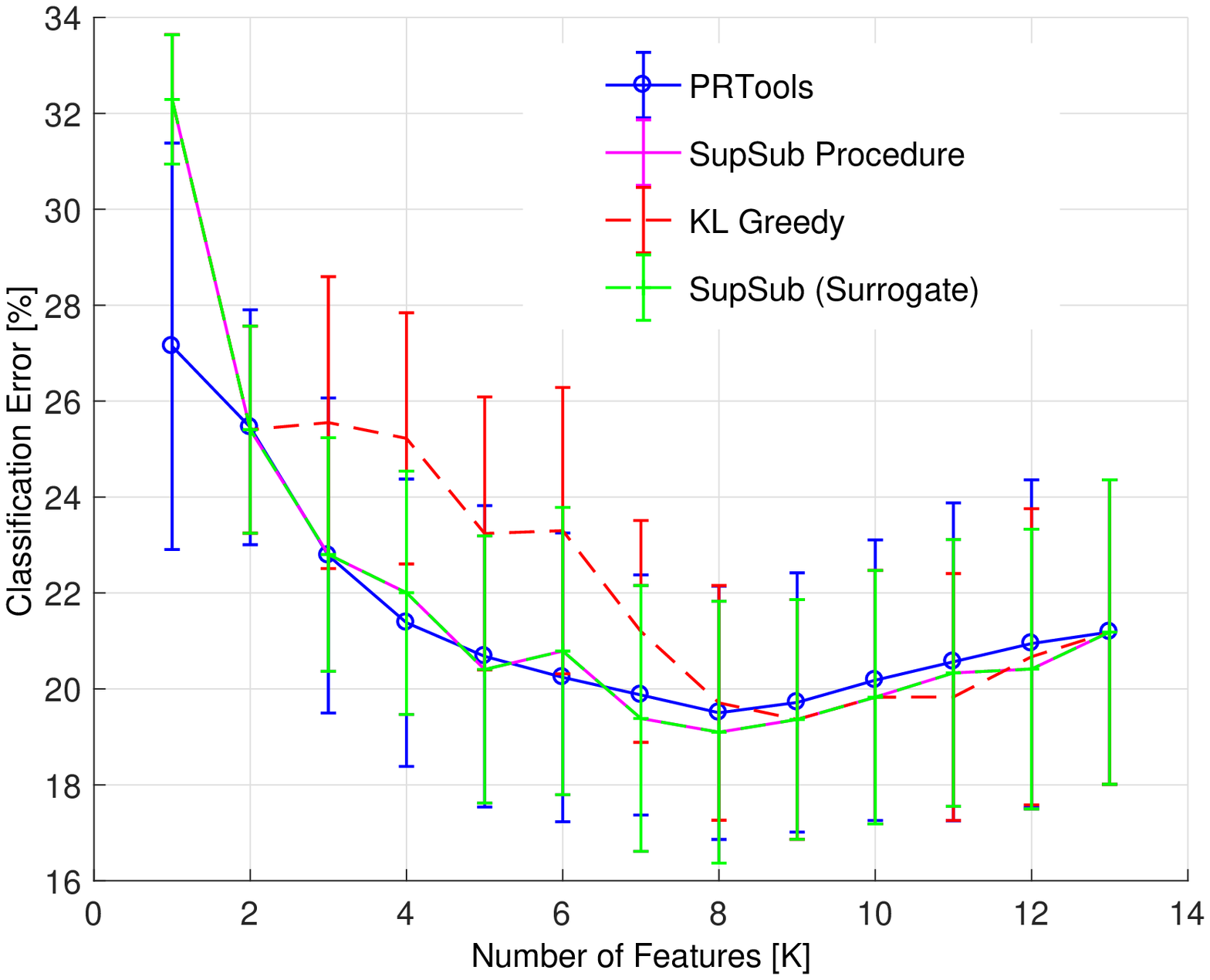}
	\caption{}
	\label{fig:classSVM}
	\end{subfigure}
	\caption{(a) Classification soft error, weighted with class priors, when using QDC for a Gaussian binary classification problem. (b) Classification error for SVMs trained using different feature selection methods for the Heart-Cleveland data set.}
\end{figure*}

\subsection*{Real Dataset Example}

As a second example, we use the \emph{Cleveland Heart Disease Data Set}~\cite{heart} in which a set of $76$ attributes from $303$ patients are reported describing the presence of a heart disease. Due to the nature of the data, only $14$ of the reported attributes are used as features, e.g., id number, social security number, and similar attributes are omitted. In the original dataset, the presence of heart disease is described by an integer number in the range $\{0,1,\ldots,4\}$, however in this scenario we consider a binary hypothesis test in which the label $l=0$ represents a \emph{healthy heart} and the labels $l\geq1$ represent a patient with any kind of heart disease. For further information of the complete dataset the reader is referred to the related online repository~\cite{repoheart}. Similarly to the previous case, only $20\%$ of the data (randomly selected) is used to train the classifier selected for this problem. In this setup, the true covariance matrices for the features are considered for performing selection. That is, from the whole data set the second-order statistics for each feature, within a given class, are computed and the resulting covariance matrix is considered as the true covariance matrix for the data. The same criterion and baseline are used to perform the selection of the features from the dataset, and the results are reported over hundred random selections for the training subset. For this dataset, a support vector machine (SVM) was trained to discriminate between the healthy and unhealthy patients. In Fig.~\ref{fig:classSVM}, the average classification error, in percentage, is reported for each method with their respective $95\%$ confidence interval. From this plot it can be observed that the methods based on the SupSub procedure produce subsets of features which attain a similar performance as the baseline, i.e., the PRTool built-in function optimizing over the QDC metric. However, the method that only uses the greedy heuristic to maximize the KL divergence obtains subsets with a worse performance for a small number of features. When the number of features is close to the maximum, the three methods based on the greedy rule perform slightly better, in both mean error and error deviation, than the baseline feature selection method. Notice the convex behavior of the classification error for the SVM classifier in Fig.~\ref{fig:classSVM}. Differently from the previous example, here the dataset structure is more complex and no Gaussian distribution properly describes it. Therefore, increasing the number of features could possibly overtrain the classifier hindering its generalization capabilities. However, it is important to notice than even when Gaussianity is not granted, the maximization of the KL divergence as a metric for feature selection leads to subsets with a smaller average classification error. 


\section{Conclusions}
\label{sec:conc}
In this paper, we have considered submodular optimization for model-based sparse sampler design for Gaussian signal detection with correlated data. Differently from traditional approaches based on convex optimization, in this work we have focused on efficient methods to solve the sensor selection problem using submodular set functions. We have shown how the discrete optimization of widely used performance metrics, for both Bayesian and Neyman-Pearson settings, can be approximated and solved using the submodular optimization machinery. For Gaussian observations with common covariance and uncommon means we bounded the $\epsilon$-submodularity constant of the SNR set function, and derived a submodular surrogate based on the Schur complement for instances in which such a constant is large. We have shown that for series of practical classes of covariance matrices this surrogate leads to a performance comparable with the convex relaxation of the problem, but at a reduced computational complexity. For the case of common means and uncommon covariance, we propose to employ the SupSub procedure for maximizing the difference of submodular set functions. When the decomposition of the divergence measure into submodular functions is not straightforward, we introduce surrogate decompositions based on the Schur complement that can be evaluated efficiently. Furthermore, a series of numerical examples with both synthetic and real data demonstrate the effectiveness of the proposed methods to perform both sensor and feature selection even when the data is not Gaussian distributed.

\appendices
\section{Proof of Theorem~\ref{th:epSub}}
\label{ap:esub}
First, let us consider the following set function,
\begin{equation}
\hat{h}(\mathcal{A}) = -\bs{\theta}^{T}\bm{S}^{-1}\bigg(a^{-1}\bm{I} + a^{-1}\text{diag}(\bm{1}_{\mathcal{A}})\bigg)^{-1}\bm{S}^{-1}\bs{\theta},
\label{eq:modSNR}
\end{equation}
where an identity matrix has been introduced instead the inverse of $\bm{S}$ [cf.~\eqref{eq:snrExt}] to construct a modular set function. In addition, without loss of generality, let us assume $\Vert\bm{S}^{-1}\bs{\theta}\Vert = 1$. 

As the SNR set function can be defined as [cf.~\eqref{eq:snrExt}]
\begin{equation}
s(\mathcal{A}) := \bs{\theta}^{T}\bm{S}^{-1}\bs{\theta} + h(\mathcal{A}),
\end{equation}
where $h(\mathcal{A}):= -\bs{\theta}^{T}\bm{S}^{-1}\bigg[\bm{S}^{-1}+a^{-1}\text{diag}(\bm{1}_{\mathcal{A}})\bigg]^{-1}\bm{S}^{-1}\bs{\theta}$. We next establish the following inequalities
\begin{equation}
-\epsilon' \leq h(\mathcal{A}) - \hat{h}(\mathcal{A}) \leq \epsilon',
\label{eq:dbineq}
\end{equation}
that for an specific $\epsilon'\in\mathbb{R}_{+}$ provides the relation in~\eqref{eq:thEpSub}. Therefore, in the following we will compute the value of $\epsilon'$ for the signal-to-noise set function. We can bound the difference of the positive definite (PD) matrices that are part of the quadratic forms in the set functions in~\eqref{eq:dbineq}. That is,
\begin{equation*}
-\epsilon'\bm{I} \preceq \bigg[a^{-1}\bm{I}+a^{-1}\text{diag}(\bm{1}_{\mathcal{A}})\bigg]^{-1} - \bigg[\bm{S}^{-1} + a^{-1}\text{diag}(\bm{1}_{\mathcal{A}})\bigg]^{-1}\preceq \epsilon'\bm{I}.
\end{equation*}

Considering $\text{diag}(\bm{1}_{\mathcal{A}})=\bm{E}_{\mathcal{A}}^{T}\bm{E}_{\mathcal{A}}$, we can apply the matrix inversion lemma to expand the difference of the matrices as
\begin{equation}
\bs{\Delta} := a\bigg(\bm{I}-\frac{1}{2}\bm{E}_{\mathcal{A}}^{T}\bm{E}_{\mathcal{A}}\bigg) - \bm{S} + \bm{S}\bm{E}_{\mathcal{A}}^{T}\bigg(a\bm{I} + \bm{E}_{\mathcal{A}}\bm{SE}_{\mathcal{A}}^{T} \bigg)^{-1}\bm{E}_{\mathcal{A}}\bm{S}.
\label{eq:delta}
\end{equation}
As all terms in~\eqref{eq:delta} are PD matrices, we upper bound the previous expression by removing the negative terms in~\eqref{eq:delta}. That is,
\begin{equation}
\bs{\Delta} \preceq a\bm{I} + \bm{S}\bm{E}_{\mathcal{A}}^{T}\bigg(a\bm{I} + \bm{E}_{\mathcal{A}}\bm{SE}_{\mathcal{A}}^{T} \bigg)^{-1}\bm{E}_{\mathcal{A}}\bm{S}.
\end{equation}
Using the maximum singular value of the second matrix, the following inequality holds
\begin{eqnarray*}
\bs{\Delta} &\preceq a\bm{I} + \sigma_{\max}\bigg\{\bm{S}\bm{E}_{\mathcal{A}}^{T}\bigg(a\bm{I} + \bm{E}_{\mathcal{A}}\bm{SE}_{\mathcal{A}}^{T} \bigg)^{-1}\bm{E}_{\mathcal{A}}\bm{S}\bigg\}\bm{I}\\
&\preceq a\bm{I} + \lambda_{\min}^{-1}\bigg\{a\bm{I} + \bm{E}_{\mathcal{A}}\bm{SE}_{\mathcal{A}}^{T}\bigg\}\sigma_{\max}^{2}\bigg\{\bm{E}_{\mathcal{A}}\bm{S}\bigg\}\bm{I}\\
&\preceq a\bm{I} + a^{-1}\lambda_{\max}^{2}\{\bm{S}\}\bm{I},
\end{eqnarray*}
where the submultiplicativity and subadditivity of singular values is used in the second and third inequality, respectively. Considering that the eigenvalues of $\bs{\Sigma}$ are larger to those of $\bm{S}$ by definition, and recalling that $a = \beta\lambda_{\min}(\bs{\Sigma})$ we obtain
\begin{eqnarray}
\bs{\Delta} \preceq a\bm{I} + a^{-1}\lambda_{\max}^{2}\{\bs{\Sigma}\}\bm{I}
\preceq \bigg(a + \frac{\kappa}{\beta}\lambda_{\max}\{\bs{\Sigma}\}\bigg)\bm{I} = \epsilon'\bm{I},
\label{eq:lstchain}
\end{eqnarray}

The expression for $\epsilon'$ in~\eqref{eq:lstchain} can now be used to bound the expression in~\eqref{eq:dbineq}. Therefore, by using the fact that $\hat{h}$ is a modular function we can write the following inequalities
\begin{equation*}
\begin{split}
&h(\sset{A}\cup\{i\}) - h({\sset{A}}) - h({\sset{A}\cup\{i,j\}}) + h({\sset{A}\cup\{j\}})\\
&\leq \hat{h}({\sset{A}\cup\{i\}}) - \epsilon' -\hat{h}({\sset{A}}) -\epsilon' - \hat{h}({\sset{A}\cup\{i,j\}}) -\epsilon' +\\
& \;\;\;\;\;\;\;\;\;\;\;\;\;\;\;\hat{h}({\sset{A}\cup\{j\}}) -\epsilon' = -4\epsilon' = -\epsilon.
\end{split}
\end{equation*}

\section{Proof of Proposition~\ref{prop:mon}}
\label{ap:mono}
\textbf{Monotonicity}: Let us define the following:
$$
\bm{T} = \begin{bmatrix}
\textbf{S}^{-1}  & \textbf{S}^{-1}\bs{\theta}\\
\bs{\theta}^{T}\textbf{S}^{-1} & \bs{\theta}^{T}\textbf{S}^{-1}\bs{\theta} 
\end{bmatrix},\; \textbf{L}_{\sset{A}} = \begin{bmatrix}
a^{-1}\text{diag}(\textbf{1}_{\sset{A}})  & \textbf{0}\\
\textbf{0} & \textbf{0}
\end{bmatrix}.
$$

We can express the cost set function from~\eqref{eq:subs2} as
$
f(\sset{A}) = \log\text{det}(\textbf{T} + \textbf{L}_{\sset{A}}),
$
where we have defined $\bm{M}_{\sset{A}}\coloneqq \bm{T} + \bm{L}_{\sset{A}}$. To prove the monotonicity of the set function we need to show
$$
f(\sset{A}\cup \{i\}) - f(\sset{A}) = \log\frac{\text{det}(\textbf{M}_\sset{A}+\bm{L}_{i})}{\text{det}(\textbf{M}_\sset{A})}.
$$
Therefore, we should prove that
$
\text{det}(\textbf{M}_\sset{A}+\bm{L}_{i}) \geq \text{det}(\textbf{M}_\sset{A}).
$
This condition is implied by
$
\textbf{M}_\sset{A} + \bm{L}_{i} \succeq \textbf{M}_\sset{A},
$
as $a\geq 0$.

\textbf{Submodularity} : Let us consider the previous definitions for $\bm T$ and $\bm L_{\sset{A}}$. We need to prove that the following expression is always positive
\begin{multline*}
f({\sset{A}\cup i}) - f(\sset{A}) - f({\sset{A}\cup \{i,j\}}) + f({\sset{A}\cup j}) =\\ \log\frac{\text{det}(\textbf{M}_\sset{A}+\textbf{L}_{i})\text{det}(\textbf{M}_\sset{A}+\textbf{L}_{j})}{\text{det}(\textbf{M}_\sset{A})\text{det}(\textbf{M}_\sset{A}+\textbf{L}_{i}+\textbf{L}_{j})} \geq 0
\end{multline*}
The above inequality is equivalent to
$$
\frac{\text{det}(\textbf{M}_\sset{A}+\textbf{L}_{i})\text{det}(\textbf{M}_\sset{A}+\textbf{L}_{j})}{\text{det}(\textbf{M}_\sset{A})\text{det}(\textbf{M}_\sset{A}+\textbf{L}_{i}+\textbf{L}_{j})} \geq 1
$$
Noticing that $\textbf{L}_{i} = a^{-1}\textbf{e}_{i}\textbf{e}_{i}^{T}$ is a dyadic product, and that $\textbf{M}_\sset{A}$ and $\textbf{M}_\sset{A} + \textbf{L}_{j}$ are invertible by definition, we can apply the matrix determinant lemma and rewrite the previous expression as
$$
\frac{\text{det}(\textbf{M}_\sset{A})\text{det}(\textbf{M}_\sset{A}+\textbf{L}_{j})(1 + a^{-1}\textbf{e}_{i}^{T}\textbf{M}_\sset{A}^{-1}\textbf{e}_{i})}{\text{det}(\textbf{M}_\sset{A})\textbf{det}(\textbf{M}_\sset{A}+\textbf{L}_{j})(1+a^{-1}\textbf{e}_{i}(\textbf{M}_\sset{A}+\textbf{L}_{j})^{-1}\textbf{e}_{i})} \geq 1,
$$
leading to
$$\frac{1 + a^{-1}\textbf{e}_{i}^{T}\textbf{M}_\sset{A}^{-1}\textbf{e}_{i}}{1+a^{-1}\textbf{e}_{i}(\textbf{M}_\sset{A}+\textbf{L}_{j})^{-1}\textbf{e}_{i}} \geq 1.
$$
Finally, the inequality for the last ratio can be proven using the following property of positive definite matrices. If $\textbf{M} \succeq \textbf{N}$, then $\textbf{M}^{-1} \preceq \textbf{N}^{-1}$. Hence,
$$
a^{-1}\textbf{e}_{i}^{T}(\textbf{M}_\sset{A}^{-1} - (\textbf{M}_\sset{A} + \textbf{L}_{i})^{-1})\textbf{e}_{i} \geq 0,
$$
which is always true for ${a}\geq 0$ and due to
$
\textbf{M}_\sset{A} + \textbf{L}_{i} \succeq \textbf{M}_\sset{A}.
$
\ifCLASSOPTIONcaptionsoff
  \newpage
\fi

%
%
%
%
%





\begin{thebibliography}{1}

\bibitem{VT1}
 H. L. Van Trees, \emph{Array Processing}, Wiley, New York, 1971

\bibitem{seis}
M. Withers, et al., ``A comparison of select trigger algorithms for automated global seismic phase and event detection," \emph{Bulletin of the Seismological Society of America} vol.88, no.1, pp.95-106, 1998.

\bibitem{CR}
E. Axell, et al., ``Spectrum sensing for cognitive radio: State-of-the-art and recent advances," \emph{{IEEE} Sig. Proc. Mag.} vol.29,no.3,pp.101-116, 2012.

\bibitem{BO}
I. Traore, et al., ``Continuous Authentication Using Biometrics: Data, Models, and Metrics: Data, Models, and Metrics." \emph{IGI Global}, 2011.

\bibitem{SS0}
S.P. Chepuri, and G. Leus, ``Sparse Sensing for Statistical Inference," \emph{Foundations and Trends in Sig. Proc.}, vol.9, no.34, pp.233-368, 2016.

\bibitem{SS1}
H. Jamali-Rad, et al., ``Sparsity-aware sensor selection: Centralized and distributed algorithms," \emph{{IEEE} Sig. Proc. Lett.}, vol.21, no.2, pp.217-220, 2014.

\bibitem{SS2}
Z. Quan, et al., ``Innovations diffusion: A spatial sampling scheme for distributed estimation and detection," \emph{{IEEE} Trans. Sig. Proc.}, vol.57, no.2, pp.738-751, 2009.

\bibitem{SS3}
J. Chaoyang, et al., ``Sensor placement by maximal projection on minimum eigenspace for linear inverse problems," \emph{{IEEE} Trans. Sig. Proc.}, vol.64, no.21, pp.5595-5610, 2015.

\bibitem{gBertrand}
A. Bertrand, and M. Moonen. ``Efficient sensor subset selection and link failure response for linear MMSE signal estimation in wireless sensor networks,", \emph{Proc. of the 18th European Sig. Proc. Conference, {IEEE}}, pp.1092-1096, Aug., 2010.

\bibitem{ssBoyd}
S. Joshi, and S. Boyd. ``Sensor selection via convex optimization." \emph{{IEEE} Trans. Sig. Proc.}, vol.57, no.2, pp.451- 462, 2009.

\bibitem{ref:weakCorrelation}
S. Liu, et al., ``Sensor selection for estimation with correlated measurement noise," \emph{{IEEE} Trans. Sig. Proc.}, vol.64, no.13, pp.3509-3522, 2016.

\bibitem{SS4}
M. Yilin, et al., ``Sensor selection strategies for state estimation in energy constrained wireless sensor networks," \emph{Automatica}, vol.47, no.7, pp.1330-1338, 2011.

\bibitem{P1}
C. Yu, and P. K. Varshney, ``Sampling design for Gaussian detection problems," \emph{{IEEE} Trans. Sig. Proc.}, vol.45, no.9, pp.2328-2337, 1997.

\bibitem{P2}
S. Cambanis, and E. Masry, ``Sampling designs for the detection of signals in noise," \emph{IEEE Trans. Inf. Theory.}, vol.29, no.1, pp.83-104, 1983.

\bibitem{P3}
R. K. Bahr and J. A. Bucklew, ``Optimal sampling schemes for the Gaussian hypothesis testing problem," \emph{{IEEE} Trans. Acoust., Speech, Sig. Proc.}, vol.38, no.10, pp.1677-1686,  1990.

\bibitem{P4}
D. Bajovic, et al., ``Sensor selection for event detection in wireless sensor networks," \emph{{IEEE} Trans. Sig. Proc.}, vol.59, no.10, pp.4938-4953, 2011.

\bibitem{sub:r15}
S.P. Chepuri, and G. Leus. ``Sparse sensing for distributed detection," \emph{{IEEE} Trans. Sig. Proc.}, vol.64, no.6, pp.1446-1460, 2015.

\bibitem{sub:r10}
A. Krause, A. Singh, and C. Guestrin, ``Near-optimal sensor placements in Gaussian processes: Theory, efficient algorithms and empirical studies," \emph{Journal of Machine Learning Research}, pp.235-284, Feb 2008 .

\bibitem{Sub2}
B. Fang, D. Kempe, and R. Govindan, ``Utility based sensor selection." \emph{Proceedings of the 5th international conference on Inf. Proc. in sensor networks,} ACM, 2006.

\bibitem{Sub1}
S. Manohar, S. Banerjee , and H. Vikalo, ``Greedy sensor selection: Leveraging submodularity," \emph{49th {IEEE} Conference on Decision and Control (CDC)}, 2010.

\bibitem{ssVet}
J. Ranieri, A. Chebira, and M. Vetterli, ``Near-optimal sensor placement for linear inverse problems," \emph{{IEEE} Trans. Sig. Proc.}, vol.62, no.5, pp.1135-1146, 2014.

\bibitem{ssVet2}
D. E. Badawy, J. Ranieri, and M. Vetterli, ``Near-optimal sensor placement for signals lying in a union of subspaces," \emph{Proc. of the 22nd European Sig. Proc. Conference, {IEEE}}, 2014.

\bibitem{ssVet3}
S. Rao, S. P. Chepuri, and G. Leus, ``Greedy sensor selection for non-linear models," \emph{Proc. of the 6th International Workshop on Computational Advances in Multi-Sensor Adaptive Processing, {IEEE}}, 2015.

\bibitem{sub:shannon}
C. E. Shannon, R. G. Gallager, and E. R. Berlekamp, ``Lower bounds to error probability for coding on discrete memoryless channels. I," \emph{Information and Control}, vol.10, no.1, pp.65-103, 1967.

\bibitem{sub:bhatt}
T. Kailath, ``The divergence and Bhattacharyya distance measures in signal selection," \emph{{IEEE} Trans. Commun. Technol.},vol.15,no.1,pp.52-60,1967.

\bibitem{jdiv}
J. Harold, ``An invariant form for the prior probability in estimation problems." \emph{Proceedings of the Royal Society of London A: mathematical, physical and engineering sciences}. vol.186. no.1007, 1946.

\bibitem{sub:kullbck}
S. Kullback, \emph{Information theory and statistics}, Courier Corp., 1997.

\bibitem{matroid}
C. Gruia, et al., ``Maximizing a monotone submodular function subject to a matroid constraint," \emph{SIAM Journal on Computing}, vol.40, no.6, pp.1740-1766,  2011.

\bibitem{sub:r1}
M. C. Shewry, and H. P. Wynn, ``Maximum entropy sampling," \emph{Journal of applied statistics}, vol.14, no.2, pp.165-170, 1987.

\bibitem{sub:r2}
T.M. Cover and J. A. Thomas, \emph{Elements of information theory}, John Wiley \& Sons, 2012.

\bibitem{sub:r3}
S. Fujishige, \emph{Submodular functions and optimization}, Elsevier, 2005.

\bibitem{sub:r4}
L. Lov\'{a}sz, ``Submodular functions and convexity," \emph{The State of the Art Mathematical Programming,} Springer Berlin, pp.235-257, 1983.

\bibitem{sub:r5}
F. Bach, ``Learning with submodular functions: A convex optimization perspective," arXiv preprint arXiv:1111.6453, 2011.

\bibitem{sub:r6}
S. Iwata, et al., ``A combinatorial strongly polynomial algorithm for minimizing submodular functions," \emph{Journal of the ACM (JACM)}, vol.48, no.4, pp.761-777, 2001.

\bibitem{sub:r7}
A. Schrijver, ``A combinatorial algorithm minimizing submodular functions in strongly polynomial time," \emph{Journal of Combinatorial Theory}, Series B, vol.80, no.2, pp.346-355, 2000.

\bibitem{sub:r8}
S. Jegelka, et al., ``On fast approximate submodular minimization," \emph{Advances in Neural Information Processing Systems}, 2011.

\bibitem{gredNew1}
R. Khanna, et al., ``Scalable Greedy Feature Selection via Weak Submodularity," arXiv preprint arXiv:1703.02723, 2017.

\bibitem{gredNew2}
X. Yuan, et al., ``Newton-Type Greedy Selection Methods for $\ell_0 $-Constrained Minimization," \emph{{IEEE} Trans. Pattern Anal. Mach. Intell.}, 2017.

\bibitem{gredNew3}
G. Shulkind, et al., ``Sensor Array Design Through Submodular Optimization," arXiv preprint arXiv:1705.06616, 2017.

\bibitem{sub:r9}
G. L. Nemhauser, et al., ``An analysis of approximations for maximizing submodular set functions," \emph{Mathematical Programming}, vol.14, no.1, pp.265-294, 1978.

\bibitem{sub:r11}
M. Narasimhan, and J. A. Bilmes, ``A submodular-supermodular procedure with applications to discriminative structure learning," arXiv preprint arXiv:1207.1404, 2012.

\bibitem{sub:r12}
R. Iyer, and J. Bilmes, ``Algorithms for approximate minimization of the difference between submodular functions, with applications," arXiv preprint arXiv:1207.0560, 2012.

\bibitem{mmAlg}
Y. Sun, et al., ``Majorization-minimization algorithms in Sig. Proc., communications, and machine learning," \emph{{IEEE} Trans. Sig. Proc.}, vol.65, no.3, pp.794-816, 2017.

\bibitem{kays}
S. M. Kay, \emph{Fundamentals of statistical Sig. Proc., vol.II: Detection Theory,} Upper Saddle River, NJ: Prentice Hall, 1998.

\bibitem{sub:r13}
A. L. Yuille, and A. Rangarajan, ``The concave-convex procedure," \emph{Neural computation}, vol.15, no.4, pp.915-936,  2003.

\bibitem{sub:r14}
A. Krause, and D. Golovin, ``Submodular function maximization," \emph{Tractability: Practical Approaches to Hard Problems}, vol.3, no.19, 2012.

\bibitem{sub:r16}
S. Boyd, and L. Vandenberghe, \emph{Convex optimization}. Cambridge University press, 2004.

\bibitem{hadi}
H. Jamali-Rad, et al., ``Sparsity-aware sensor selection for correlated noise," \emph{Proc. of the 17th Int. Conf. on Inf. Fusion, {IEEE},} 2014.

\bibitem{sub:r17}
Z. Luo, et al, ``Semidefinite relaxation of quadratic optimization problems," \emph{{IEEE} Sig. Proc. Mag,} vol.27, no.3, 2010.

\bibitem{sub:lchamon1}
L. Chamon, and A. Ribeiro, ``Near-optimality of greedy set selection in the sampling of graph signals," \emph{Proc. of the Global Conference on Signal and Information Processing, {IEEE},} 2016.


\bibitem{ref:lazy}
J. Leskovec, et al., ``Cost-effective outbreak detection in networks," \emph{Proc. of the 13th ACM SIGKDD International Conference on Knowledge Discovery and Data Mining, {ACM}}, 2007.

\bibitem{ref:lazyx2}
B . Mirzasoleiman, et al., ``Lazier Than Lazy Greedy," \emph{AAAI}, 2015.


\bibitem{PRqdc}
S. Theodoridis, et al., \emph{Introduction to pattern recognition: a Matlab approach}, Academic Press, 2010.


\bibitem{prtools}
R.P.W. Duin, et al., \emph{PRTools5, A Matlab Toolbox for Pattern Recognition}, Delft University of Technology, 2017.

\bibitem{heart}
M.D. A. Janosi, \emph{Heart-Cleveland dataset from the Hungarian Institute of Cardiology}. Budapest. 

\bibitem{repoheart}
https://archive.ics.uci.edu/ml/datasets/Heart+Disease

\bibitem{subbnd}
S. Jegelka, et al., ``Submodularity beyond submodular energies: coupling edges in graph cuts," \emph{Proc. of the Conf. on Comp. Vision and Pattern Recognition, {IEEE}}, 2011

\end{thebibliography}
\end{document}